\documentclass[twocolumn]{svjour3}
\usepackage{graphicx}
\usepackage{amssymb}
\usepackage{amsfonts}
\usepackage{amsmath}
\usepackage{indentfirst}
\usepackage{float}

\RequirePackage{fix-cm}
\smartqed  
\begin{document}

\title{Weak Values and Modular Variables From a Quantum Phase Space
Perspective }
\author{Augusto C. Lobo        \and
		Yakir Aharonov		\and
		Jeff Tollaksen		\and        
		Elizabeth M. Berrigan \and        
        Clyffe A. Ribeiro   
        }

\institute{Augusto C. Lobo \at
              Departamento de F\'{i}sica, Instituto de Ci\^{e}ncias Exatas e Biol\'{o}gicas, Universidade Federal de
Ouro Preto, Ouro Preto, MG 35400-000, Brazil. \\
              \email{lobo@iceb.ufop.br}             \\
              \emph{Present address:} Institute for Quantum Studies, Chapman University, 1 University Dr, Orange, CA 92866, USA.              \and
			Yakir Aharonov   \at
              Tel Aviv University, School of Physics and Astronomy, Tel Aviv 69978, Israel.    \\
			  Schmid College of Science and Technology, Chapman University, 1 University Dr, Orange, CA 92866, USA.  \\
              \email{aharonov@chapman.edu}  \and          
			Jeff Tollaksen \at
              Institute for Quantum Studies, Chapman University, 1 University Dr, Orange, CA 92866, USA.  \\     
              \email{tollakse@chapman.edu}   \and 
            Elizabeth M. Berrigan \at
              Institute for Quantum Studies, Chapman University, 1 University Dr, Orange, CA 92866, USA.  \\     
              \email{berri104@mail.chapman.edu}   \and                        
           Clyffe A. Ribeiro \at
              Departamento de F\'{i}sica, Instituto de Ci\^{e}ncias Exatas, Universidade Federal de Minas Gerais, Belo Horizonte, MG 31270-901, Brazil. \\
              \email{clyffe@fisica.ufmg.br} \\
              \emph{Present address:} Institute for Quantum Studies, Chapman University, 1 University Dr, Orange, CA 92866, USA.   
}

\date{Received: date / Accepted: date}

\dedication{In memory of our colleague and friend, Maria Carolina Nemes, who was a shining example of excellence and team effort.}

\maketitle

\begin{abstract}
We address two major conceptual developments introduced by Aharonov and
collaborators through a \textit{quantum phase space} approach: the concept
of \textit{modular variables} devised to explain the phenomena of quantum
dynamical non-locality and the \textit{two-state formalism} for Quantum
Mechanics which is a retrocausal time-symmetric interpretation of quantum
physics which led to the discovery of \textit{weak values.} We propose that
a quantum phase space structure underlies these profound physical insights
in a unifying manner. For this, we briefly review the Weyl-Wigner and the
coherent state formalisms as well as the inherent symplectic structures of
quantum projective spaces in order to gain a deeper understanding of the
weak value concept.

We also review Schwinger's finite quantum kinematics so that we may apply
this discrete formalism to understand Aharonov's modular variable concept in
a different manner that has been proposed before in the literature. We
discuss why we believe that this\ is indeed the correct kinematic framework
for the modular variable concept and how this may shine some light on the
physical distinction between quantum dynamical non-locality and the
kinematic non-locality, generally associated with entangled quantum systems.

 \PACS{03.65.Ta \and 42.50.-p \and 03.65.Vf \and 03.65.Aa }
\end{abstract}

\keywords{weak values \and modular variables \and phase space}





\section{Introduction}

\label{intro}

There are two major conceptual developments introduced by Aharonov and
collaborators which we believe provide the underpinnings for a more
fundamental understanding of Quantum Mechanics (QM). The \textit{first}
is the concept of \textit{modular variables} which was devised to
characterize the kind of dynamical non-locality that become widely
recognized when he, together with David Bohm, shocked the world of physics
in the late fifties with the introduction of the so called Aharonov-Bohm
(AB) effect \cite{aharonov_bohm1959}. Their paper discussed a situation
where the wave functions that represent the electrons suffer a phase shift from the magnetic field of a
solenoid even without having had any direct contact with the field. In \cite%
{aharonovpendletonpetersen}, Aharonov, Pendleton and Petersen introduced the
concept of \textit{modular variables} to explain this kind of \textit{%
non-local} quantum effect as the modular momentum exchange between particles
and fields. This is contrary to the usual view where the AB phenomenon is
explained by a \textit{local} interaction between the particles and field
potentials, even if the potentials are considered somewhat unphysical
because they are defined only up to some gauge transformation. Even the
interference effects of a single particle passing through a slit must be
described by the this non-local dynamical effect. Imagine a particle beam
diffracting through a two-slit apparatus (with a distance $L$ between the
slits). The Schr\"{o}dinger wave-like picture is capable of explaining the
interference by analogy with classical wave-like interferometry. But this
picture is misleading and it is not capable of giving a full quantum
mechanical explanation of the underlying phenomena. One can lower the
intensity of the beam until only \textit{one} single particle arrives at the
apparatus at each moment. One may choose to open or close one of the slits
until the last instant. How does one of the slits \textquotedblleft know"
about the other one instantaneously? To answer this question, one can devise
the following simplified mathematical model: Consider the problem as a
two-dimensional setting in the $x$-$y$ plane with the beam in the positive $y
$ direction and the apparatus with the slits in the $x$ direction. We
consider only the $x$ direction where the particle \textquotedblleft feels"
a potential $\phi (x)$.  It is not difficult to see that the translation
operator $\hat{V}_{L}=e^{i\hat{P}L}$ obeys a \textit{non-local} equation of
motion in the Heisenberg picture:%
\begin{equation}
\frac{d}{dt}\left( \hat{V}_{L}\right) =i\left[ \phi (\hat{Q}),\hat{V}_{L}%
\right] =i\left( \phi (\hat{Q})-\phi (\hat{Q}+L)\right) \hat{V}_{L}
\label{dynamical non-local equation}
\end{equation}

This kind of \textit{dynamical} non-locality seems to be of a fundamentally
different physical nature from that of the \textit{kinematic} type of
non-locality that occurs for entangled systems that are spatially appart as
for instance, the paradigmatic EPR pair of particles. One of the main
results of this paper is to advance our understanding of the differences
between these two types of non-locality by using Schwinger%
\'{}%
s finite kinematic formalism in order to present an alternative approach for
the construction of the Hilbert spaces for modular variables. Though the AB
effect is nowadays discussed in any undergraduate text-book of QM, according
to Aharonov, the \textit{true} reason for these strange phenomena, where
electric and magnetic fields seem to be capable of exerting
\textquotedblleft action at a distance" upon charged quantum particles, has
not been generally appreciated. Aharonov developed an intuition on this
concept through the \textit{Heisenberg picture} of QM\ and he has argued in
a convincing manner that it is easier to think of this quite subtle issue
through Heisenberg's picture\ rather than Schr\"{o}dinger's wave-like
picture \cite{aharonov2003book}. In fact, in the early seventies, Heisenberg
himself was very pleased to learn from Aharonov how to describe the concept
of quantum interference based on his own formulation of quantum physics
instead of the usual particle-wave (Schr\"{o}dinger-like) approach \cite%
{aharonov2013}.

The \textit{second} concept is the advancement of the \textit{two-state
formalism} for QM \cite{AharonovBergmannLebowitz}. This is a radically
different (time-symmetric) view of QM\ that has also been developed since
the sixties and implies a notion of retrocausality in quantum physics which
led to the discovery of \textit{weak values} \cite{AharonovAlbertVaidman}. The development of the concept
of weak values and weak measurements has spanned a number of important
theoretical and experimental breakthroughs from quantum counterfactual
investigations to high precision metrology and even led to some purely
mathematical ideas such as the concept of superoscillations \cite%
{aharonov2002revisiting,dixon2009,hosten2008,berry2010typical}.

We recall that new points of view concerning a certain subject are always
welcome. It is common historical knowledge that the establishment of bridges
between distinct disciplines is usually a very fruitful enterprise for both
subjects. This interplay has brought us (from at least Newton and Galileo to
Einstein and Minkowski, passing through Euler, Lagrange, Hamilton, Maxwell
and many others) a wonderful multitude of results where mathematical
structures are discovered by contemplating natures wonders and physical
theories are guessed from deep and beautiful mathematics. In this paper, we
approach the many deep physical insights of Aharonov with an eye towards the
quantum phase space structure which we propose may underlie them in a
unifying \ manner.

The structure of this paper is as follows: in section 2, we review the
mathematical structure of classical phase space as the geometric notion that
\textquotedblleft classical mechanics is symplectic
geometry\textquotedblright\ \cite{arnold,abraham,sudarshan2006} so that the
focus is the natural symplectic structure of cotangent bundles of
configuration manifolds of classical particles restricted to holonomic
constraints. This symplectic geometric structure appears consistently in
almost all the topics covered in this paper. In section 3, we present an
intrinsic formulation for the operator algebra of the WW Transform (similar
to Dirac's Bra and Ket notation) that is necessary to construct this
formalism in a more compact and elegant way than the usual manner. In
section 4, we discuss the natural symplectic structure that projective
spaces of finite dimensional Hilbert spaces inherit from their Hermitian
structure and for this reason are somewhat structurally analogous to
classical phase spaces. In section 5, we briefly review the coherent state
concept in terms of the WW basis mainly to set the stage for a discussion of
the quantum transforms that implement linear area preserving maps on the
phase plane. In section 6, we apply some of these ideas to a phase space
study of the weak value concept, which has become an important theoretical
and experimental tool in modern investigations of quantum physics. We look
at the phase space of the \textit{measuring apparatus} that performs the
weak value measurement of some quantum system. In section 7, we approach the
weak measurements and von Neumann pre-measurements by the opposite approach:
in particular, we look instead to the geometric structure of the \textit{%
measured system}. In section 8, we review the \textit{geometry of
deterministic measurements} and of \textit{completely uncertain operators}.
We also discuss some examples that can be seen as \textit{partially}
deterministic operators. This is the starting point in order to one fully
appreciate the Heisenberg picture approach that Aharonov has championed for
the description of quantum interference phenomena. Finally, in section 9, we
introduce the modular variable concept. We review the \textit{original}
approach due to Aharonov and collaborators for the construction of an
explicit Hilbert space (an appropriate \textquotedblleft qudit space" for
modular variables). We then review Schwinger's finite quantum kinematics in
order to apply this discrete formalism to understand Aharonov's modular
variable concept in a different manner that has been proposed until now in
the literature. In the original approach, the total quantum space can be
thought as the \textit{direct orthogonal sum} between the modular variable
qudit space and \textquotedblleft the rest" of the infinite dimensional
quantum space. Within the Schwinger formalism, we introduce a new proposal
based on the \textit{tensor product} between the modular variable qudit
space and the rest of the quantum space. While there are some
number-theoretic subtleties in order to perform this construction, we
discuss some of the advantages of considering this Schwinger finite quantum
kinematic-based proposal as the correct \textquotedblleft kinematic arena"
for the modular variable concept. In section 10 we finalize with some
concluding remarks and we also set stage for further research.

\section{The Mathematical Structure of Classical Phase Space}

(The main references for this section are \cite{arnold} and \cite{abraham})

What is meant here by a \textquotedblleft classical
structure\textquotedblright\ is the totality of the mathematical formalism
associated with conservative (Hamiltonian) dynamical systems, \textit{i.e.}:

\begin{enumerate}
\item An even-dimensional differential manifold $M$ \newline
$(\dim M=2n)$;

\item A canonical coordinate system $(q^{i},p_{i})$ \newline
$(i=1,2,\ldots,n)$ over $M$, where $n$ is the number of degrees of freedom
of the system;

\item A non-degenerate closed 2-form $%
\Omega
$ on $M$.
\end{enumerate}

In other words: $d\Omega=0$ and the non-degeneracy meaning that: $%
\Omega(V,W)=0\Rightarrow V=0$ or $W=0$, $\forall\text{ } V, W \in T_{p}M$
and $\forall\text{ } p\in M$ and where $TM$ is the tangent bundle of $M$ and 
$T_{p}M$ is the tangent space (fiber) of $M$ at $p\in M$. A differential
manifold with the above structure is called symplectic. A theorem due to
Darboux guarantees that given a point $p\in M$, there is always an open set
of $M$ that contains $p$ and that admits a pair of canonical coordinates ($%
q^{i},p_{i}$ ) such that $\Omega=-dp_{i}\wedge dq^{i}$ (we use henceforth
the sum convention unless we explicitly state the contrary). In Classical
Newtonian Mechanics, the above structure arises naturally as the cotangent
bundle of a configuration manifold $Q$ (of dimension $n$) of $m$ particles
submitted to $3m-n$ holonomic constraints: $M=T^{\ast}Q$. In fact, any
cotangent bundle has a natural 1-form $\theta=p_{i}dq^{i}$, where $q^{i}$
are the coordinates of $Q$ and $p^{i}$ are the coordinates of the co-vectors
of $T^{\ast}Q$. It is easy to verify that indeed the 2-form defined by $%
\Omega=-d\theta$ satisfies automatically the conditions (1-3). Yet not all
symplectic manifolds are cotangent bundles.

An example that will be of great importance for us later is that of quantum
projective spaces. We may define a dynamical structure on a symplectic
manifold (the phase space) by introducing a real-value Hamiltonian function $%
H:M\rightarrow\mathbb{R}$ which maps each point in $M$ to a definite energy
value. We can also define a symplectic gradient $X_{H}$ of $H$ (through $%
\Omega
$) by the following relation $%
\Omega
(X_{H},Y)=dH(Y)$ for any vector field $Y$ defined over $M$. In canonical
coordinates, the vector field $X_{H}$ is given by%
\begin{equation}
X_{H}=\frac{\partial H}{\partial p_{i}}\frac{\partial}{\partial q^{i}}-\frac{%
\partial H}{\partial q^{i}}\frac{\partial}{\partial p_{i}}
\end{equation}

The motion of the system is given by the integral curves of $X_{H}$ and the
first order ODE's associated to them are the well-known Hamilton equations:%
\begin{equation}
\dot{q}^{i}=\frac{\partial H}{\partial p_{i}}\qquad\text{and}\qquad\dot{p}%
_{i}=-\frac{\partial H}{\partial q^{i}}
\end{equation}

A classical observable is any well-behaved real function defined on $M$. For
example, the observables $Q^{i}$ and $P_{i}$ are respectively the $i$-th
components of the generalized position and momentum%
\begin{equation*}
\begin{array}{c}
Q^{i}:M\rightarrow\mathbb{R} \\ 
(q^{i},p_{i})\longrightarrow q^{i}%
\end{array}
\qquad\text{and}\qquad%
\begin{array}{c}
P_{i}:M\rightarrow\mathbb{R} \\ 
(q^{i},p_{i})\longrightarrow p_{i}%
\end{array}%
\end{equation*}

The Poisson brackets of two arbitrary observables $f$ and $g$ can then be
defined as:%
\begin{equation}
\{f,g\}=%
\Omega
(X_{f},X_{g})=\frac{\partial g}{\partial p_{i}}\frac{\partial f}{\partial
q^{i}}-\frac{\partial f}{\partial p_{i}}\frac{\partial g}{\partial q^{i}},
\end{equation}
and the temporal evolution of a general observable \newline
$O(q^{i},p_{i})$ can be written as%
\begin{equation}
\frac{dO}{dt}=\{O,H\}  \label{classical time evolution of observables}
\end{equation}

The different formulations of Quantum Mechanics have shown from the very
beginning a close relation with the structure of classical analytical
mechanics as one can infer from common denominations of quantum physics as
like, for example, the \textquotedblleft Hamiltonian
operator\textquotedblright\ for the generator of time displacements and the
formal analogies of equations like the Heisenberg equation and its classical
analog given by \eqref{classical time evolution of observables} (we use
henceforth $\hbar=1$ units)%
\begin{equation}
i\frac{d\hat{O}}{dt}=[\hat{O},\hat{H}]
\end{equation}
for the evolution of a quantum observable $\hat{O}$ and a classical
observable $O\left( q^{i},p_{i}\right) $. Since then, some different
formalisms that relate both structures have been introduced with a number of
applications in physics. In the next sections, we review some of these
formalisms such as the Weyl-Wigner transform theory and coherent state
theory. In the final sections we show how these structures are intimately
related to the concepts of weak values and modular variables.

One of the most fundamental differences between quantum physics and
classical physics is that the first admits the possibilities of both
discrete and continuous observables while the second admits only continuous
quantities. In the next section we will discuss only continuous phase spaces
and in the following sections we will address \textquotedblleft
finite\textquotedblright\ or \textquotedblleft discrete\textquotedblright\
phase spaces when we introduce the Schwinger formalism as a tool towards a
deeper understanding of the concept of modular variables.

\section{The Weyl-Wigner Formalism}

(The main references for this section are \cite{degroot} and \cite{sakurai})

\subsection{The Fourier Transform Operator}

Consider a quantum system defined by the motion of a single non-relativistic
particle in one dimension. Let $\left\{ \left\vert q(x)\right\rangle
\right\} $ and $\left\{ \left\vert p(x)\right\rangle \right\} $ $%
(-\infty<x<\infty)$ represent respectively the quantum eigenkets of the
position and momentum observables. In other words, we have:%
\begin{equation}
\hat{Q}\left\vert q(x)\right\rangle =x\left\vert q(x)\right\rangle \qquad%
\text{and}\qquad\hat{P}\left\vert p(x)\right\rangle =x\left\vert
p(x)\right\rangle
\end{equation}
where $\hat{Q}$ and $\hat{P}$ are the position and momentum operators,
respectively. It is important to notice here that we use a slightly
different notation than the usual choice. For instance, the ket $\left\vert
q(x)\right\rangle $ is an eigenvector of $\hat{Q}$ with eigenvalue $x$. That
is, we distinguish the \textquotedblleft kind\textquotedblright\ of
eigenvector (position or momentum) from its eigenvalue. This is different
from the more common notation where $\left\vert q\right\rangle $ and $%
\left\vert p\right\rangle $ represent both the \textit{type} of eigenket
(respectively position and momentum in this case) and also the \textit{%
eigenvalue} in the sense that $\hat{Q}\left\vert q\right\rangle =q\left\vert
q\right\rangle $ and $\hat{P}\left\vert p\right\rangle =p\left\vert
p\right\rangle $. As we shall see shortly, our choice of notation will allow
us to write some equations in a more compact and elegant form. We shall
designate the vector space generated by these basis as $W^{(\infty)}$. This
is clearly not a Hilbert space since $W^{(\infty)}$ accommodates
\textquotedblleft generalized vectors\textquotedblright\ as $\left\vert
q(x)\right\rangle $ and $\left\vert p(x)\right\rangle $ that have
\textquotedblleft infinite norm\textquotedblright. A rigorous foundation for
this construction can be given within the so called rigged vector space
formalism (see \cite{arnobohm} for more details on this issue). Later on, we
will present a natural construction of these spaces as a continuous
heuristic limit of analogous well-defined finite dimensional spaces
originally due to Schwinger.

Let $\hat{V}_{\xi}$ and $\hat{U}_{\eta}$ $(-\infty<\xi,\eta<\infty)$ be a
pair of unitary operators that implement the one parameter abelian group of
translations on respectively the position and momentum basis in the sense
that%
\begin{equation}
\hat{V}_{\xi}\left\vert q(x)\right\rangle =\left\vert q(x-\xi)\right\rangle
\quad\text{and}\quad\hat{U}_{\eta}\left\vert p(x)\right\rangle =\left\vert
p(x+\eta)\right\rangle
\end{equation}
The hermitian generators of $\hat{V}_{\xi}$ and $\hat{U}_{\eta}$ are the
momentum and position observables, respectively:%
\begin{equation}
\hat{V}_{\xi}=e^{i\xi\hat{P}}\qquad\text{and}\qquad\hat{U}_{\eta}=e^{i\eta 
\hat{Q}}  \label{translation operators}
\end{equation}
so that 
\begin{equation}
\hat{V}_{\xi}\left\vert p(x)\right\rangle =e^{i\xi x}\left\vert
p(x)\right\rangle \quad\text{and}\quad\hat{U}_{\eta}\left\vert
q(x)\right\rangle =e^{i\eta x}\left\vert q(x)\right\rangle
\label{eigenvalue eq. for translators}
\end{equation}
The translation operators also obey the so called Weyl relation:%
\begin{equation}
\hat{V}_{\xi}\hat{U}_{\eta}=\hat{U}_{\eta}\hat{V}_{\xi}e^{i\eta\xi}
\label{continuous weyl commutation relation}
\end{equation}
which can be thought as an exponentiated version of the familiar Heisenberg
relation:%
\begin{equation}
\lbrack\hat{Q},\hat{P}]=i\hat{I}
\end{equation}

The basis $\left\vert q(x)\right\rangle $ and $\left\vert p(x)\right\rangle $
are both complete and normalized in the sense that%
\begin{equation}
\hat{I}=\int_{-\infty}^{+\infty}dx\left\vert q(x)\right\rangle \left\langle
q(x)\right\vert =\int_{-\infty}^{+\infty}dx\left\vert p(x)\right\rangle
\left\langle p(x)\right\vert
\end{equation}
and%
\begin{equation}
\left\langle q(x)\right\vert q(x\prime)\rangle=\left\langle p(x)\right\vert
p(x\prime)\rangle=\delta(x-x\prime)
\end{equation}
The overlap between the position and momentum eigenstate is given by the
well-known plane wave function:%
\begin{equation}
\left\langle q(x)\right\vert p(x\prime)\rangle=\frac{1}{\sqrt{2\pi}}%
e^{ixx\prime}  \label{plane wave function}
\end{equation}

Any arbitrary abstract state $\left\vert \psi\right\rangle $ has a
\textquotedblleft position basis wave function\textquotedblright given by $%
\left\langle q(x)\right\vert \psi\rangle$ and a \textquotedblleft momentum
basis wave function\textquotedblright\ given by $\left\langle
p(x)\right\vert \psi\rangle$. The relation between both is given by the
Fourier transform operator which is a unitary operator that takes one basis
to another as%
\begin{equation}
\hat{F}\left\vert q(x)\right\rangle =\left\vert p(x)\right\rangle
\label{def. 1 of the fourier operator}
\end{equation}
This operator can be defined in a more elegant and compact manner as:%
\begin{equation}
\hat{F}=\int_{-\infty}^{+\infty}dx\left\vert p(x)\right\rangle \left\langle
q(x)\right\vert
\end{equation}
Note that the above equation could \textit{not} have been written in such
way within the conventional notation. We also have that%
\begin{equation*}
\hat{F}^{\dag}\left\vert p(x)\right\rangle =\left\vert q(x)\right\rangle
\end{equation*}
and that $\hat{F}$ is unitary, which means that 
\begin{equation*}
\hat{F}^{\dag}\hat{F}=\hat{I}
\end{equation*}
The Fourier transform of $\psi(x)=\left\langle q(x)\right\vert \psi\rangle$
is \text{ }then $\left\langle q(x)\right\vert \hat{F}^{\dag}\left\vert
\psi\right\rangle =\left\langle p(x)\right\vert \psi\rangle=\tilde{\psi}(x)$%
, where we introduced here the tilde psi-function $\tilde{\psi}(x)$ as the
Fourier transform of $\psi(x)$. The effect of the inverse transform Fourier
operator over the position basis is given by%
\begin{equation*}
\hat{F}^{\dag}\left\vert q(x)\right\rangle
=\int_{-\infty}^{+\infty}dx\prime\left\vert q(x\prime)\right\rangle
\left\langle p(x\prime)\right\vert q(x)\rangle
\end{equation*}
From the symmetries of \eqref{plane wave function} it is easy to see that $%
\left\langle p(x\prime)\right\vert q(x)\rangle=\left\langle q(x\prime
)\right\vert p(-x)\rangle$ so that:%
\begin{equation*}
\hat{F}^{\dag}|q(x)\rangle=\int_{-\infty}^{+\infty}dx\prime\left\vert
q(x\prime)\right\rangle \left\langle q(x\prime)\right\vert p(-x)\rangle
=\left\vert p(-x)\right\rangle
\end{equation*}
Analogously, we may write%
\begin{eqnarray*}
\hat{F}^{2}\left\vert q(x)\right\rangle &=&\hat{F}\left\vert
p(x)\right\rangle =\int_{-\infty}^{+\infty}dx\prime\left\vert
p(x\prime)\right\rangle \left\langle q(x\prime)\right\vert p(x)\rangle 
\notag \\
&=&\int_{-\infty}^{+\infty }dx\prime\left\vert p(x\prime)\right\rangle
\left\langle p(x\prime)\right\vert q(-x)\rangle=\left\vert q(-x)\right\rangle
\end{eqnarray*}
This implies that $\hat{F}^{2}$ is nothing else but the \textit{space
inversion operator} in $W^{(\infty)}$. Analogously one can deduce in a
similar manner that 
\begin{equation*}
\hat{F}^{2}\left\vert p(x)\right\rangle =\left\vert p(-x)\right\rangle
\end{equation*}
This also implies that%
\begin{equation*}
\hat{F}^{4}=\hat{I}
\end{equation*}
which means that the spectrum of the Fourier operator $\hat{F}$ is
constituted by the \textit{fourth roots of unity}. We shall return to this
issue in a later section. In a similar manner it is not difficult to show
that%
\begin{equation}
\hat{F}^{\dag}\hat{U}_{\xi}\hat{F}=\hat{V}_{\xi}^{\dag}\qquad\text{and}\qquad%
\hat{F}^{\dag}\hat{V}_{\xi}\hat{F}=\hat{U}_{\xi}  \label{relations i}
\end{equation}
\ and 
\begin{equation}
\hat{F}\hat{V}_{\xi}\hat{F}^{\dag}=\hat{U}_{\xi}^{\dag}\qquad\text{and}\qquad%
\hat{F}\hat{U}_{\xi}\hat{F}^{\dag}=\hat{V}_{\xi}  \label{relations ii}
\end{equation}
which implies the following two very important relations%
\begin{equation}
\hat{F}^{2}\hat{U}_{\xi}\hat{F}^{2}=\hat{U}_{\xi}^{\dag}\qquad\text{and}%
\qquad\hat{F}^{2}\hat{V}_{\xi}\hat{F}^{2}=\hat{V}_{\xi}^{\dag}
\label{relations iii}
\end{equation}

\subsection{An Intrinsic Formulation for the Weyl-Wigner Operator and
Transform}

We shall designate the space of linear operators over $W^{(\infty)}$ as $%
T_{1}^{1}(W^{(\infty)})$ and denote the \textquotedblleft
vectors\textquotedblright (the elements) $\hat{A},\hat{B},\hat{C}$,\ldots\
of $T_{1}^{1}(W^{(\infty)})$ respectively by $|\hat{A})$, \newline
$|\hat{B})$,$|\hat{C})$,\ldots Following Schwinger \cite{schwinger}, we
introduce the well-known hermitian inner product, also known as the
Hilbert-Schmidt inner product, in $T_{1}^{1}(W^{(\infty)})$ given by:%
\begin{equation}
(\hat{A}|\hat{B})=tr(\hat{A}^{\dag}\hat{B})
\label{Hilbert-Schmidt inner product}
\end{equation}
The space $T_{1}^{1}(W^{(\infty)})$ has an additional algebraic structure
which turns it into an operator algebra. In fact, we have the following
product $|\hat{A}).|\hat{B})=|\hat{A}.\hat{B})$, where $\hat{A}.\hat{B}$ is
the usual operator product between $|\hat{A})$ and $|\hat{B})$. In this way
we may identify the elements of $T_{1}^{1}(W^{(\infty)})$ with the elements
of $T_{1}^{1}(T_{1}^{1}(W^{(\infty)}))$ in an unique manner through the
obvious inclusion map%
\begin{align}
i:T_{1}^{1}(W^{(\infty)})&\rightarrow T_{1}^{1}(T_{1}^{1}(W^{(\infty)})) 
\notag \\
\hat{A}\equiv|\hat{A})&\mapsto\breve{A},  \notag
\end{align}
such that $\breve{A}|\hat{B})=|\hat{A}.\hat{B})$. We will usually allow
ourselves a slight abuse of language by dismissing any explicit mention to
this inclusion map. For instance, consider the following identifications of
the Identity Operator: $\hat{I}\equiv|\hat{I})\overset{i}{\equiv}\breve{I}$
. In this way one may define a set of operators $\hat{X}(\alpha)\equiv|\hat {%
X}(\alpha))$ to be \textit{complete} in $T_{1}^{1}(W^{(\infty)})$, where $%
\alpha$ is a variable that takes values in some appropriate
\textquotedblleft index set\textquotedblright, if%
\begin{equation}
\int d\alpha|\hat{X}(\alpha))(\hat{X}(\alpha)|=\breve{I}\equiv\hat{I}\equiv|%
\hat{I}),  \label{completeness of basis in operator space}
\end{equation}
where $d\alpha$ is an appropriate measure in the index set. Orthonormality
can be written as%
\begin{equation}
(\hat{X}(\alpha)|\hat{X}(\beta))=tr(\hat{X}^{\dag}(\alpha)\hat{X}%
(\beta))=\delta(\alpha-\beta)  \label{basis orthonormality in operator space}
\end{equation}

As an example, consider $|q(x)\rangle\langle p(y)|\equiv|x,y)$ with $x,y\in%
\mathbb{R}
^{2}$. This set of $%
\mathbb{R}
^{2}$-valued operators is indeed orthonormal in the sense that 
\begin{equation}
(x,y|x\prime,y\prime)=\delta(x-x\prime)\delta(y-y\prime)
\end{equation}
It is not difficult to prove that $(x,y|\hat{A})=\langle q(x)|\hat {A}%
|p(y)\rangle$ and $\int\int dxdy|x,y)(x,y|=\breve{I}$. These results can be
used to prove another very convenient completeness relation for a set $\{%
\hat{X}(\alpha)\}$ as%
\begin{equation*}
\int d\alpha\hat{X}^{\dag}(\alpha)\hat{A}\hat{X}(\alpha)=(tr\hat{A})\hat {I}%
,\quad\forall\text{ }\hat{A}\in T_{1}^{1}(W^{(\infty)})
\end{equation*}

We can now define a set of operators (the so called Weyl-Wigner operators)
also parameterized by the phase space plane $(x,y)\in%
\mathbb{R}
^{2}$ in the following intrinsic manner%
\begin{equation}
\hat{\Delta}(x,y)=2\hat{V}_{y}^{\dag}\hat{U}_{2x}\hat{V}_{y}^{\dag}\hat{F}%
^{2}  \label{definition of the WW  operator basis}
\end{equation}
With aid of \eqref{relations iii}, one can easily prove the following
important properties:

\begin{itemize}
\item hermiticity: 
\begin{equation}
\hat{\Delta}^{\dag}(x,y)=\hat{\Delta}(x,y)
\label{hermiticity of the WW basis}
\end{equation}

\item unit trace: 
\begin{equation}
tr(\hat{\Delta}(x,y))=(\hat{\Delta}(x,y)|\hat{I})=1
\label{unit trace of the WW operator basis}
\end{equation}

\item orthonormality of ${\frac{1}{\sqrt{2\pi}}\hat{\Delta}(x,y)}$: 
\begin{equation}
(\hat{\Delta}(x,y)|\hat{\Delta}(x\prime,y\prime))=2\pi\delta(x-x\prime
)\delta(y-y\prime)  \label{orthonormality of the WW operator basis}
\end{equation}

\item 1st form for the completeness of $\frac{1}{\sqrt{2\pi}}\hat{\Delta }%
(x,y)$: 
\begin{equation}
\frac{1}{2\pi}\int\int dxdy|\hat{\Delta}(x,y))(\hat{\Delta}(x,y)|=\hat{I}
\label{1st form of completeness of WW basis}
\end{equation}

\item 2nd form of the completeness of $\frac{1}{\sqrt{2\pi}}\hat{\Delta}%
(x,y) $: 
\begin{equation}
\frac{1}{2\pi}\int\int dxdy\hat{\Delta}(x,y)\hat{A}\hat{\Delta}(x,y)=(tr%
\hat {A})\hat{I},\quad\forall\hat{A}
\label{2nd form of completeness of WW basis}
\end{equation}
\end{itemize}

Some other properties (not so well-known) for the Weyl-Wigner operators are
the following:%
\begin{equation}
\hat{\Delta}^{2}(x,y)=4\hat{I}  \label{p1}
\end{equation}%
\begin{equation}
\hat{F}^{2}\hat{\Delta}(x,y)\hat{F}^{2}=\hat{\Delta}(-x,-y)  \label{p2}
\end{equation}%
\begin{equation}
\hat{\Delta}(0,0)=2\hat{F}^{2}  \label{p3}
\end{equation}

At this point, one can define the Weyl-Wigner Transform of an arbitrary
Operator as%
\begin{equation}
W\{\hat{A}\}(x,y)=(\hat{\Delta}(x,y)|\hat{A})=a(x,y)
\label{definition of the WW transform}
\end{equation}%
The transform of $\hat{A}$ is in general a \textit{complex} function $a(x,y)$
of the phase plane, but if $\hat{A}$ is \textit{hermitian}, then $a(x,y)$ is
clearly \textit{real}-valued. One is tempted to see this transform as a map
between quantum observables to \textquotedblleft classical
observables\textquotedblright\ in some sense. In fact, as we shall soon see,
there is a certain sense where in the $\hbar \rightarrow 0$ limit, one can
see that $a(x,y)$ goes indeed to the expected classical observable. Wigner
introduced this transform to map density operators of mixed states to
classical probability densities over phase space, but the fact is that these
densities obey all axioms for a true probability distribution on phase space
with the exception of \textit{positivity}. In this way, the negativity of
the transform of a density operator signals for a departure of classicality
of the mixed state or some kind of measure of \textquotedblleft
quanticity\textquotedblright\ of the state.

\subsection{The Classical Limit}

The inner product between two operators $\hat{A}$ and $\hat{B}$ can be
written in terms of their WW transforms as%
\begin{equation}
(\hat{A}| \hat{B})=\frac{1}{2\pi}\int dxdy\overline{a}(x,y)b(x,y),
\end{equation}
where $\overline{a}(x,y)$ is the c.c. of $a(x,y)$. The transform of a
product of two arbitrary operators can be seen, after a tedious calculation,
as%
\begin{equation}
(\hat{\Delta}(x,y)| \hat{A}\hat{B})=a(x,y)\exp\Big [-\frac{i}{2}(\frac{%
\overleftarrow{\partial}}{\partial x}\frac{\overrightarrow{\partial}}{%
\partial y}-\frac{\overleftarrow{\partial}}{\partial y}\frac {%
\overrightarrow{\partial}}{\partial x})\Big ]b(x,y),  \label{star-product}
\end{equation}
where the arrows above the partial derivative operators indicate
\textquotedblleft which\ function\textquotedblright\ it operates on. With
this result, it is not difficult to guess the very suggestive form of the WW
transform of the commutator of two operators and (we momentarily restore $%
\hbar$ here explicitly) the following fact that the \textquotedblleft
classical limit\textquotedblright\ can be understood as 
\begin{align}
&-2i\underset{\hbar\rightarrow0}{\lim}\Big ( \frac{1}{\hbar}(\hat{\Delta }%
(x,y)|[\hat{A},\hat{B}])\Big )  \notag \\
& =2\underset{\hbar\rightarrow0}{\lim }\Big(\frac{1}{\hbar}a(x,y)\sin\Big [ 
\frac{\hbar}{2}\Big (\frac {\overleftarrow{\partial}}{\partial x}\frac{%
\overrightarrow{\partial}}{\partial y}-\frac{\overleftarrow{\partial}}{%
\partial y}\frac {\overrightarrow{\partial}}{\partial x}\Big )\Big ]b(x,y)%
\Big )  \notag \\
& =\{a(x,y),b(x,y)\}  \label{classical limit}
\end{align}
One recognizes the last term of the above equation as the Poisson bracket of
the classical observables $a(x,y)$ and $b(x,y)$. The operation defined by %
\eqref{star-product} establishes a \textit{non-commutative algebra} over the
real functions defined on phase space. It is also known as the star-product
and it forms the basis of the non-commutative geometric point of view for
quantization.

\subsection{The Issue of the Number of Degrees of Freedom}

The extension of the continuous WW formalism to two or more degrees of
freedom is straightforward: For instance, let $|q(x)\rangle$ and $%
|q(y)\rangle$ be complete position eigenbasis for the $x$ and $y$
directions, respectively. Then, a point $\vec{r}$ of the plane is clearly
represented by the tensor product state $|q(\vec{r})\rangle=|q(x)\rangle%
\otimes|q(y)\rangle$. The 2D translation operator 
\begin{equation}
\hat{V}_{\vec{\xi}}=e^{i\vec{P}.\vec{\xi}} =e^{i\hat{P}_{x}\xi_{x}}\otimes
e^{i\hat{P}_{y}\xi_{y}}
\end{equation}
acts upon the $|q(\vec{r})\rangle$ basis by the expected manner as $\hat {V}%
_{\vec{\xi}}|q(\vec{r})\rangle=|q(\vec{r}-\vec{\xi})\rangle$. This can be
clearly carried out for any number of degrees of freedom in the same way.
For finite dimensional quantum spaces, things are \textit{not} as simple as
we shall see further ahead.

\section{The Geometry of Quantum Mechanics}

(The main references for this section are \cite{arnold}, \cite{page} and 
\cite{anandan1990})

\subsection{The Projective Space $\mathbb{CP}(N)$}

Let us consider the $(N+1)$-dimensional quantum space $W^{(N+1)}$. This is a
complex vector space together with an \textit{anti-linear} map between $%
W^{(N+1)}$ and its dual $\overline{W}^{(N+1)}$:%
\begin{align*}
\langle\text{ }\rangle:W^{(N+1)} & \rightarrow\overline{W}^{(N+1)} \\
|\psi\rangle & \mapsto\langle\psi|=(|\psi\rangle)^{\dag}
\end{align*}
that takes each ket state to its associated bra state through the familiar
\textquotedblleft dagger\textquotedblright\ operation. The inner product
between $|\psi_{1}\rangle$ and $|\psi_{2}\rangle$ can then be defined as the
natural action of $\langle\psi_{1}|$ over $|\psi_{2}\rangle$ in the usual
Dirac notation $\left\langle \psi_{1}\right\vert \psi_{2}\rangle$. Following 
\cite{arnold}, one can define Euclidean and symplectic metric structures on $%
W^{(N+1)}$ through the relations%
\begin{align}
&G\left( |\psi_{1}\rangle,|\psi_{2}\rangle\right) =\operatorname{Re}\left\langle
\psi_{1}\right\vert \psi_{2}\rangle \quad\text{and}\quad  \notag \\
&\Omega\left( |\psi_{1}\rangle,|\psi_{2}\rangle\right) =\operatorname{Im}\left\langle
\psi_{1}\right\vert \psi_{2}\rangle
\end{align}

From fundamental postulates of Quantum Mechanics, it is clear that a state
vector $|\psi\rangle\in W^{(N+1)}$ is physically indistinguishable from the
state vector $\lambda|\psi\rangle$ (for arbitrary $\lambda\in\mathbb{C}$).
Thus, the true physical space of states of the theory is the so called
\textquotedblleft space of rays\textquotedblright\ or the complex projective
space $\mathbb{CP}(N)$ defined by the quotient of $W^{(N+1)}$ by the above
physically motivated equivalence relation.

Given an orthonormal basis $\{|u_{\sigma}\rangle\}$ $\sigma=0,1,2$,\ldots,
then a general state vector can be expanded as $|\psi\rangle=|u_{\sigma
}\rangle z^{\sigma}$. One can map this state to a sphere $S^{2N+1}$ with
squared radius given by $r^{2}=\bar{z}_{\sigma}z^{\sigma}$. We introduce
projective coordinates 
\begin{equation}
\xi^{i}=z^{i}/z^{0}  \label{projective coordinates}
\end{equation}
on $\mathbb{CP}(N)$ so that 
\begin{equation}
z^{0}=re^{i\phi}/(1+\bar{\xi}_{j}\xi^{j})^{1/2}\qquad\text{with}\qquad
i=1,\ldots,n,
\end{equation}
where $\phi$ is an arbitrary phase factor. The Euclidean metric in $%
W^{(N+1)} $ (seen here as a $(2n+2)$-dimensional \textit{real} vector space)
can be written as 
\begin{equation}
dl^{2}\left( W^{(N+1)}\right) =d\bar{z}_{\sigma}dz^{%
\sigma}=dr^{2}+dS^{2}(S^{2N+1})
\end{equation}
where%
\begin{equation*}
dS^{2}(S^{2N+1})=(d\phi-A)^{2}+ds^{2}(\mathbb{CP}(N))
\end{equation*}
is the squared distance element over the space of normalized vectors and the
one-form 
\begin{equation*}
A=\frac{i}{2}\frac{(\bar{\xi}_{i}d\xi^{i}-d\bar{\xi}_{i}\xi^{i})}{(1+\bar{%
\xi }_{i}\xi^{i})}
\end{equation*}
is the well known abelian connection of the $U(1)$ bundle over $\mathbb{CP}%
(N)$ \cite{page}. The metric $ds^{2}(\mathbb{CP}(N))$ over the space of rays
in projective coordinates is given explicitly by:%
\begin{equation*}
ds^{2}(\mathbb{CP}(N))=\frac{i}{2}\frac{[(1+\bar{\xi}_{i}\xi^{i})\delta
_{k}^{j}-\bar{\xi}_{k}\xi^{j})]}{(1+\bar{\xi}_{i}\xi^{i})^{2}}d\bar{\xi}%
_{j}d\xi^{k}
\end{equation*}
and $%
\Omega
=dA$ is a 2-form defined over $\mathbb{CP}(N)$ which makes it a symplectic
manifold. These quantum symplectic spaces are physically very different in
origin from their Newtonian counterparts.

The Newtonian phase spaces are cotangent bundles over some configuration
manifold and as such they are always non-compact manifolds, while projective
spaces are compact for finite $N$. For instance, $\mathbb{CP}(1)$ can be
identified with a 2-sphere. The quantum projective spaces also have an
additional structure given by their Riemannian metric, which is absent in
the Newtonian case. Both these structures are compatible in a precise sense
that makes these complex projective spaces examples of what is called a 
\textit{Kahler} manifold \cite{arnold}. A more natural and intuitive
pictorial representation of these structures can be seen easily in figure %
\ref{figsspaceofstates}. 
\begin{figure}[H]%
\centering
\includegraphics[scale=.4]%
{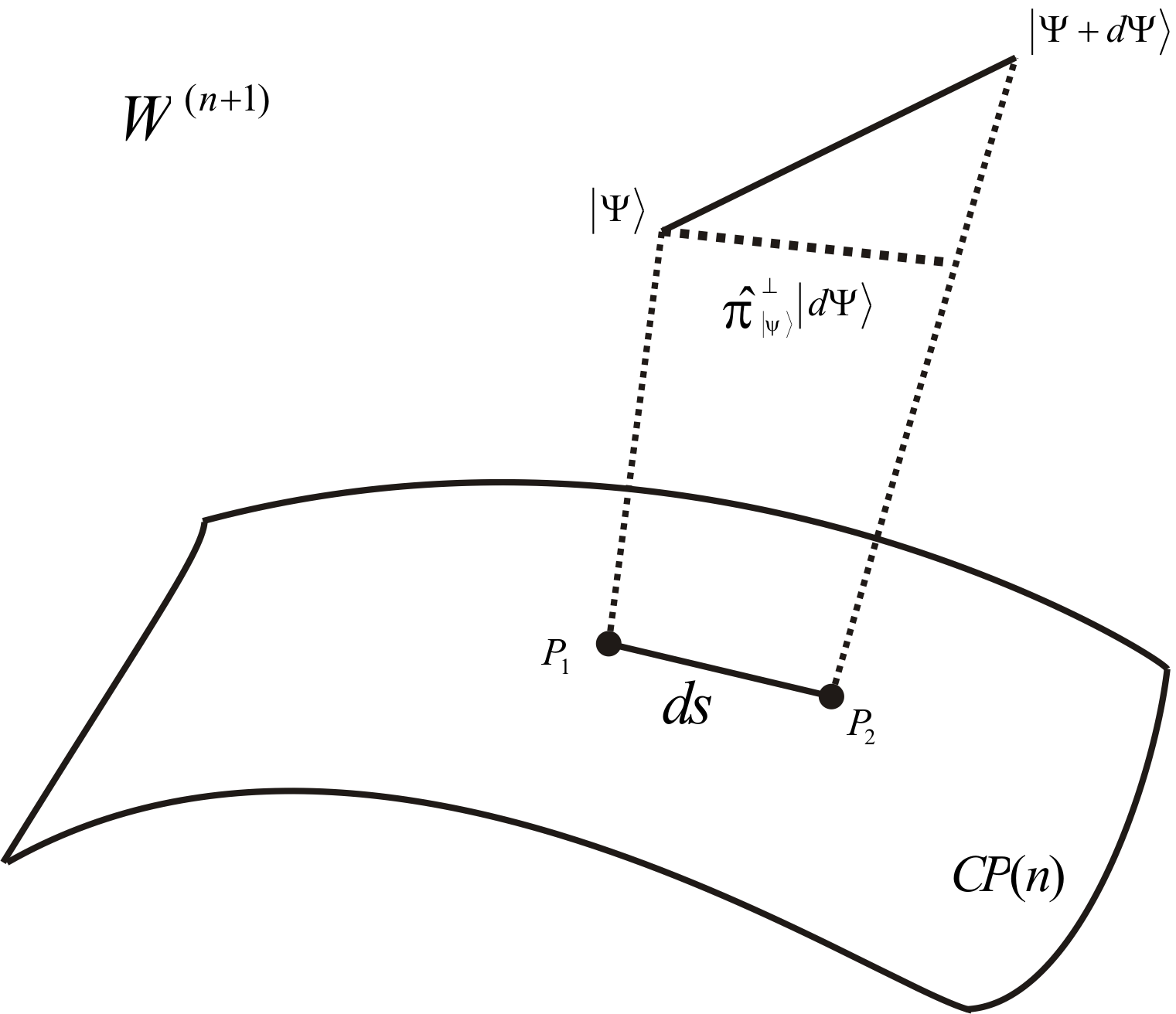}%
\caption{{ Pictorial representation of the quantum space
of states.}}\label{figsspaceofstates}
\end{figure}


The points $P_{1}$ and $P_{2}\in\mathbb{CP}(N)$ are the projections
respectively from two infinitesimally nearby normalized state vectors $%
\left\vert \psi\right\rangle $ and $\left\vert \psi+d\psi\right\rangle $. It
is natural to define then, the squared distance between $P_{1}$ and $P_{2}$
as the projection of $\left\vert d\psi\right\rangle $ in the subspace
orthogonal to $\left\vert \psi\right\rangle $, that is, the projection given
by the projective operator $\hat{\pi}_{\left\vert \psi\right\rangle }=\hat {I%
}-\left\vert \psi\right\rangle \left\langle \psi\right\vert $. It is then
easy to see that%
\begin{equation}
ds^{2}(\mathbb{CP}(N))=\left\langle d\psi\right\vert
d\psi\rangle-\left\langle d\psi\right\vert \psi\rangle\left\langle
\psi\right\vert d\psi\rangle  \label{distance element in CP(N)}
\end{equation}
which is an elegant and coordinate independent manner to express the metric
over $\mathbb{CP}(N)$. The time-evolution of a physical state in $\mathbb{CP}%
(N)$ is given by the projection of the linear Schr\"{o}dinger evolution in $%
W^{(N+1)}$ which results in a Hamiltonian (classical-like) evolution in the
space of rays.

\section{Coherent States and the Quantum Symplectic Group}

(The main references for this section are \cite{gilmore,perelomov,klauder,combescure})

\subsection{The Spectrum of the Fourier Transform Operator}

Let $\hat{a}$ and $\hat{a}^{\dag}$\ be respectively the usual annihilation
and creation operators defined as \cite{sakurai}%
\begin{equation}
\hat{a}=\frac{1}{\sqrt{2}}(\hat{Q}+i\hat{P})\qquad\text{and}\qquad\hat {a}%
^{\dag}=\frac{1}{\sqrt{2}}(\hat{Q}-i\hat{P})
\end{equation}
It follows immediately from the Heisenberg commutation relations that $[\hat{%
a},\hat{a}^{\dag}]=\hat{I}$, and if we define also the hermitian number
operator $\hat{N}=\hat{a}^{\dag}\hat{a}$, it is not difficult to derive the
well-known relations:%
\begin{equation}
\lbrack\hat{N},\hat{a}]=-\hat{a}\qquad\text{and}\qquad\lbrack\hat{N},\hat {a}%
^{\dag}\ ]=\hat{a}^{\dag}
\end{equation}
which implies the also well-known spectrum $\hat{N}\left\vert n\right\rangle
=n\left\vert n\right\rangle$$(n=0,1,2,\ldots)$\ of the number operator and
the \textquotedblleft up and down-the-ladder\textquotedblright\ action for $%
\hat{a}$ and $\hat{a}^{\dag}$:%
\begin{equation}
\hat{a}\left\vert n\right\rangle =\sqrt{n}\left\vert n\right\rangle \quad%
\text{and}\quad\hat{a}^{\dag}\left\vert n\right\rangle =\sqrt {n+1}%
\left\vert n+1\right\rangle
\end{equation}
and also where the ground-state or vacuum state $\left\vert 0\right\rangle $
is annihilated by $\hat{a}$:%
\begin{equation}
\hat{a}\left\vert 0\right\rangle =0
\end{equation}
The above algebraic equation for the vacuum state can be written in the
position and momentum basis as the following differential equations: 
\begin{align}
\langle q(x)|\hat{a}|0\rangle&=\frac{1}{\sqrt{2}}(x+\frac{d}{dx}%
)\left\langle q(x)\right\vert 0\rangle=\langle p(x)|\hat{a}|0\rangle  \notag
\\
&=\frac{1}{\sqrt{2}}(x+\frac{d}{dx})\left\langle p(x)\right\vert 0\rangle=0
\end{align}
Which gives us the normalized Gaussian functions 
\begin{equation}
\left\langle q(x)\right\vert 0\rangle=\left\langle p(x)\right\vert
0\rangle=\pi^{-1/4}e^{-x^{2}/2}
\end{equation}
From equation \eqref{def. 1 of the fourier operator} we conclude that the
vacuum state is a fixed point of the Fourier Transform operator 
\begin{equation}
\hat{F}|0\rangle=|0\rangle  \label{the vacuum state as eigenstate of F}
\end{equation}
The infinitesimal versions of equations \eqref{relations i} and %
\eqref{relations ii} are simply%
\begin{equation*}
\hat{F}^{\dag}\hat{a}\hat{F}=i\hat{a}\qquad\text{and}\qquad\hat{F}^{\dag}%
\hat{a}^{\dag}\hat{F}=-i\hat{a}^{\dag}
\end{equation*}
Which implies that 
\begin{equation*}
\lbrack\hat{F},\hat{N}]=0
\end{equation*}
Thus, the Fourier transform operator \textit{commutes} with the number
operator so they necessarily share the same eigenstates. From the fixed
point condition of the vacuum state together with the equations above, it is
not difficult to see that%
\begin{equation*}
\hat{F}|n\rangle=(i)^{n}|n\rangle\qquad\text{and}\qquad\hat{F}=e^{i\frac{\pi 
}{2}\hat{N}},
\end{equation*}
which implies that $\hat{N}$ is the hermitian generator of the Fourier
transform. In fact, at this point, it is a quite obvious move to introduce a
complete rotation operator in phase space known as the \textit{Fractional
Fourier Operator} given by \cite{Namias} 
\begin{equation*}
\hat{F}_{\theta}=e^{i\theta\hat{N}}
\end{equation*}
The Fourier transform operator is then recognized as a special case for $%
\theta=\pi/2$.

\subsection{The Quantum Linear Symplectic Transforms}

It is also natural to extend the Fractional Fourier transform to a complete
set of quantum symplectic transforms, those unitary transformation in $%
W^{(\infty)}$ that implement a representation of the group of area
preserving linear maps of the classical phase plane. This is the non-abelian 
$SL(2,%
\mathbb{R}
)$ group. Probably the best way to visualize this is through the
identification of the phase plane with the complex plane via the standard
complex-valued coherent states defined by the following change of variables: 
$z=(1/\surd2)(q+ip)$ and by defining the coherent states as 
\begin{equation}
|z\rangle=\hat{D}[z]|0\rangle\text{, with}\quad\hat{D}[z]=\frac{1}{2}\hat{%
\Delta}(z/2)\hat{F}^{2}=e^{(z\hat{a}^{\dag}-z\hat{a})}
\label{def. of coherent states}
\end{equation}
The $\hat{D}[z]$ operator is known as the displacement operator and the
well-known \textquotedblleft over-completeness\textquotedblright\ of the
coherent state representation follows from the completeness of the $\hat{%
\Delta}(z)$ basis. The overlap between two arbitrary coherent states is
given by%
\begin{equation}
\left\langle z\right\vert z\prime\rangle=\left\langle p,q\right\vert
p\prime,q\prime\rangle=e^{-1/4\left[ (p-p\prime)^{2}+(q-q\prime)^{2}\right]
}e^{i(p\prime q-pq\prime)/2}  \label{overlap of coherent states}
\end{equation}
Note above, the symplectic phase proportional to the \textit{area} in the
phase plane defined by the vectors $(p,q)$ and $(p\prime,q\prime)$. It is
also not difficult to show that indeed 
\begin{equation}
\hat{F}_{\theta}|z\rangle=e^{i\theta\hat{N}}|z\rangle=|e^{i\theta}z\rangle,
\label{rotation in Q.phase space with fractional fourier transform}
\end{equation}
which is a direct manner to represent the Fractional Fourier transform for
arbitrary $\theta$. Since the generator of rotations is quadratic in the
canonical observables $\hat{Q}$ and $\hat{P}$, one may try to write down 
\textit{all} possible quadratic operators in these variables: $\hat{Q}^{2},%
\hat{P}^{2},\hat{Q}\hat{P}$ and $\hat{P}\hat{Q}$, but the last two are
obviously \textit{non-Hermitian} so we could change them to the following
(Hermitian) linear combinations: $\hat{Q}\hat{P}+\hat{P}\hat{Q}$ and $i(\hat{%
Q}\hat{P}-\hat{P}\hat{Q})$. The last one is proportional to the identity
operator because of the Heisenberg commutation relation, so this leaves us
with \textit{three} linear independent operators that we choose as%
\begin{equation}
\hat{H}_{0}=\frac{1}{2}(\hat{Q}^{2}+\hat{P}^{2})=\hat{N}+\frac{1}{2}\hat {I}=%
\hat{a}^{\dag}\hat{a}+\frac{1}{2}\hat{I}
\label{generator of rotation in phase space}
\end{equation}%
\begin{equation}
\hat{g}=\frac{1}{2}\{\hat{Q},\hat{P}\}=\frac{1}{2}(\hat{Q}\hat{P}+\hat{P}%
\hat{Q})=\frac{i}{2}[(\hat{a}^{\dag})^{2}-\hat{a}^{2}]
\label{generator of scale transf. in phase space}
\end{equation}%
\begin{equation}
\hat{k}=\frac{1}{2}(\hat{Q}^{2}-\hat{P}^{2})=\frac{1}{2}[(\hat{a}%
^{\dag})^{2}+\hat{a}^{2}]
\label{generator of hyperbolic transf. in phase space}
\end{equation}
These three generators implement in $W^{(\infty)}$, the algebra $sl(2,%
\mathbb{R}
)$ of $SL(2,%
\mathbb{R}
)$. The $\hat{g}$ operator is nothing but the squeezing generator from
quantum optics \cite{vedral}. Indeed, the scale operator%
\begin{equation}
\hat{S}_{\xi}=e^{i\hat{g}\ln\xi}  \label{scale transform}
\end{equation}
generated by $\hat{g}$ act upon the position and momentum basis respectively
as 
\begin{equation*}
\hat{S}_{\xi}|q(x)\rangle=\sqrt{\xi}|q(\xi x)\rangle\quad\text{and}\quad 
\hat{S}_{\xi}|p(x)\rangle=\frac{1}{\sqrt{\xi}}|p(x/\xi)\rangle
\end{equation*}

The $\hat{k}$ operator generates \textit{hyperbolic rotations}, that is,
linear transformations of the plane that preserve an indefinite metric. It
takes the hyperbola $x^{2}-y^{2}=1$ into itself in an analogous way that the
Euclidean rotation takes the circle $x^{2}+y^{2}=1$ into itself. $SL(2,%
\mathbb{R}
)$ is the Lie Group of all area preserving linear transformations of the
plane, so we can identify it with the $2\times2$ real matrices with unit
determinant. Since $\det e^{X}=e^{trX}$, we can also identify the algebra $%
sl(2,%
\mathbb{R}
)$ with all $2\times2$ real matrices with null trace. Thus, it is natural to
make the following choice for a basis in this algebra:%
\begin{equation*}
\begin{array}{c}
\hat{X}_{1}=\hat{\sigma}_{1} \\ 
\hat{X}_{2}=i\hat{\sigma}_{2} \\ 
\hat{X}_{3}=\hat{\sigma}_{3},%
\end{array}%
\end{equation*}
where we have written (for practical purposes) the elements of the algebra
in terms of the well-known Pauli matrices. This is very adequate because
physicists are familiar with the commutation relations of the Pauli matrices
since they form a two-dimensional representation of the angular momentum
algebra and we can make use of these algebraic relations to completely
characterize the $sl(2,%
\mathbb{R}
)$ algebra. In fact, the mapping described by the table below relates these
algebra elements directly to the algebra of their representation carried on $%
W^{(\infty)}$):%
\begin{equation}
\begin{array}{cc}
\text{{\small {Generators of }}}sl(2,\mathbb{R}) & \text{{\small {\
Generators of the representation}}} \\ 
\hat{X}_{1}=\hat{\sigma}_{1} & -i\hat{k} \\ 
\hat{X}_{2}=i\hat{\sigma}_{2} & -i\hat{H}_{0} \\ 
\hat{X}_{3}=\hat{\sigma}_{3} & -i\hat{g}%
\end{array}%
\end{equation}
With a bit of work, it is not difficult to convince oneself that these
mapped elements indeed obey identical commutation relations.

\section{Weak Values and Quantum Mechanics in Phase Space}

The weak value of a quantum system was introduced by Aharonov, Albert and
Vaidman (A.A.V.) in \cite{AharonovAlbertVaidman} based on the two-state
formalism for Quantum Mechanics \cite%
{AharonovBergmannLebowitz,AharonovVaidman} and it generalizes the concept of
an expectation value for a given observable. Let the initial state of a
product space $W=W_{S}\otimes W_{M}$ be given by the product state $%
|\Psi\rangle =|\alpha\rangle\otimes|\phi_{(I)}\rangle$ where $|\alpha\rangle$
is the pre-selected state of the system and $|\phi_{(I)}\rangle$ is the
initial state of the apparatus. Suppose further that a \textquotedblleft
weak Hamiltonian\textquotedblright\ governs the interaction between the
system and the measuring apparatus as:%
\begin{equation}
\hat{H}_{int}=\epsilon\delta(t-t_{0})\hat{O}\otimes\hat{P}\qquad
(\epsilon\rightarrow0),  \label{hamiltonian for weak measurement}
\end{equation}
where $\hat{O}$ is an arbitrary observable to be measured in the system.
After the ideal instantaneous interaction that models this von Neumann
(weak) measurement \cite{vonNeumann}, suppose we post-select a certain final
state $|\beta\rangle$ of the system after performing a strong measurement on
it. In this case, the final state of the apparatus is clearly given by%
\begin{align}
|\phi_{(F)}\rangle&=(\left\langle \beta\right\vert \otimes\hat{I}%
)e^{-i\epsilon\hat{O}\otimes\hat{P}}(|\alpha\rangle\otimes|\phi_{(I)}\rangle)
\notag \\
&\approx\left\langle \beta\right\vert \alpha\rangle(1-i\epsilon
O_{w})|\phi_{(I)}\rangle,
\end{align}
where%
\begin{equation}
O_{w}=\frac{\left\langle \beta\right\vert \hat{O}|\alpha\rangle}{%
\left\langle \beta\right\vert \alpha\rangle}  \label{def. of weak value}
\end{equation}
is the weak value of the observable $\hat{O}$ for these particular chosen
pre and post-selected states. Note that the weak value $O_{w}$ of the
observable is, in general, an arbitrary complex number. Note also that,
though $|\phi_{(I)}\rangle$ is a normalized state, the $|\phi_{(F)}\rangle$
state vector in general, is \textit{not} normalized. In the original
formulation of (A.A.V.), the momentum $\hat{P}$ acts upon the measuring
system, implementing a small translation of the initial wave function in the
position basis, but which can be measured from the mean value of the results
of a large series of identical experiments. That is, the expectation value
of the position operator $\hat{Q}$ over a large ensemble with the same pre
and post selected states. One can generalize this procedure \cite{jozsa} by
taking an arbitrary operator $\hat{M}$ in the place of $\hat{Q}$ as the
observable of $W_{(M)}$ to be measured. In this case, the usual expectation
values of $\hat{M}$ for the initial and final states $|\phi_{(I)}\rangle$
and $|\phi_{(F)}\rangle$ are respectively:%
\begin{align}
&\langle\hat{M}\rangle_{(I)}=\langle\phi_{(I)}|\hat{M}|\phi_{(I)}\rangle
\qquad\text{and}  \notag \\
&\langle\hat{M}\rangle_{(F)}=\frac{\langle\phi_{(F)}|\hat{M}%
|\phi_{(F)}\rangle}{\langle\phi_{(F)}|\phi_{(F)}\rangle}
\end{align}
and the difference between these expectation values (the shift of $\hat{M}$)
to first order in $\epsilon$ is given in general by \cite{jozsa}:%
\begin{align}
\Delta\hat{M}&=\langle\hat{M}\rangle_{(F)}-\langle\hat{M}\rangle_{(I)} 
\notag \\
&=\epsilon\lbrack(\operatorname{Im}(O_{w})(\langle\phi_{(I)}|\{\hat{M},\hat {P}%
\}|\phi_{(I)}\rangle-  \notag \\
& -2\langle\phi_{(I)}|\hat{P}|\phi_{(I)}\rangle\langle\phi_{(I)}|\hat{M}%
|\phi_{(I)}\rangle)-  \notag \\
&-i\operatorname{Re}(O_{w})\langle\phi_{(I)}|[\hat{M},\hat{P}]|\phi_{(I)}\rangle]
\end{align}

For the choice $\hat{M}=\hat{Q}$ and also by using the Heisenberg picture
for the time evolution and by choosing the most general Hamiltonian $\hat{H}=%
\frac{1}{2m}\hat{P}^{2}+V(\hat{Q})$ for the measuring system one can derive
the following shift (also using the $sl(2,%
\mathbb{R}
)$ algebra and the Heisenberg commutation relation):%
\begin{equation}
\Delta\hat{Q}=\epsilon\lbrack\operatorname{Re}(O_{w})+m\operatorname{Im}(O_{w})\frac{d}{dt}%
(\delta_{|\phi_{(I)}\rangle}^{2}\hat{Q})],
\label{weak displacement of position}
\end{equation}
where $\delta_{|\phi_{(I)}\rangle}\hat{Q}$ is the uncertainty of the initial
state $|\phi_{(I)}\rangle$ and analogously for $\hat{M}=\hat{P}$, one
arrives at%
\begin{equation}
\Delta\hat{P}=2\epsilon\operatorname{Im}(O_{w})(\delta_{|\phi_{(I)}\rangle }^{2}\hat{%
P})
\end{equation}
This result is clearly \textit{asymmetric} because of the choice of the
translation generator $\hat{P}$ in the interaction Hamiltonian. Note also
that from the above equations one can see that it is impossible to extract
the real and imaginary values of the weak value with the measurement of $%
\Delta\hat{Q}$ only, because both of these numbers are absorbed in a same
real number. It is necessary to measure $\Delta\hat{P}$ (besides knowing the
values of $\frac {d}{dt}(\delta_{|\phi_{(I)}\rangle}^{2}\hat{Q})$ and $%
\delta_{|\phi _{(I)}\rangle}^{2}\hat{P}$. There is no reason why one should
need to choose $\hat{P}$ or $\hat{Q}$ in the weak measurement Hamiltonian.
One may choose any of the symplectic generators making use of the full
symmetry of the $SL(2,%
\mathbb{R}
)$ group. The $\hat{P}$ and $\hat{Q}$ operators generate translations in
phase space, but one can implement any area preserving transformation in the
plane by also using observables that are quadratic in the momentum and
position observables. By making use of the freedom of choice of an arbitrary
initial state vector $|\phi_{(i)}\rangle$ one can choose also an interaction
Hamiltonian of the following form:%
\begin{equation}
\hat{H}_{int}=\epsilon\delta(t-t_{0})\hat{O}\otimes\hat{R}\qquad
(\epsilon\rightarrow0),
\end{equation}
where $\hat{R}$ is \textit{any} element of the algebra $sl(2,%
\mathbb{R}
)$, so it is the generator of an arbitrary symplectic transform of the
measuring system. In this way the generalized $\Delta\hat{M}$ shift in these
conditions is given by:%
\begin{align*}
\Delta\hat{M}= & \epsilon\big [ (\operatorname{Im}(O_{w})(\langle\phi _{(I)}\vert\{%
\hat{M},\hat{R}\}|\phi_{_{(I)}}\rangle- \\
& -2\langle\phi_{_{(I)}}|\hat{R}|\phi_{_{(I)}}\rangle\langle\phi_{_{(I)}}|%
\hat{M}|\phi_{_{(I)}}\rangle)-  \notag \\
&-i\operatorname{Re}(O_{w})\langle\phi_{_{(I)}}| [\hat{M},\hat{R}]|\phi_{_{(I)}}%
\rangle\big ]
\end{align*}
By making the choice $\hat{M}=\hat{R}$, one arrives at:%
\begin{equation*}
\Delta\hat{R}=2\epsilon\operatorname{Im}(O_{w})(\delta_{|\phi_{_{(I)}}\rangle }^{2}%
\hat{R})
\end{equation*}

For the second observable, we could choose any observable that \textit{does
not} commute with $\hat{R}$. This is because the main idea is to choose a
``conjugate" variable to $\hat{R}$ in a similar way that occurs with the $(%
\hat{Q},\hat{P})$ pair. So one obvious choice is to pick the number operator 
$\hat{N}$ in the place of $\hat{R}$. Since $\hat{N}$ is the generator of
Euclidean rotations in phase space, the annihilator operator $\hat{a}$ seems
a natural candidate operator (though not hermitian) to go along with $\hat{N}
$. With this choice of $\hat{M}=\hat{a}$ it is not difficult to calculate
the shift for the annihilator operator:%
\begin{align}
\Delta\hat{a}&=\epsilon\lbrack-iO_{w}\langle\phi_{_{(I)}}\vert\hat{a}%
|\phi_{_{(I)}}\rangle +2\operatorname{Im}(O_{w})(\langle\phi_{_{(I)}}\vert \hat{N}%
\hat{a}|\phi_{_{(I)}}\rangle-  \notag \\
&-\langle\phi_{_{(I)}}\vert\hat{N}|\phi_{_{(I)}}\rangle\langle\phi_{_{(I)}}%
\vert\hat{a}|\phi_{_{(I)}}\rangle)]
\end{align}

In most models of weak measurements, the initial state of the measuring
system is chosen to be a Gaussian state and the weak interaction promotes a
small translation of its center. In realistic quantum optical
implementations of the measuring system, it is reasonable to choose the
initial state of the system as a coherent state $|\phi_{_{(I)}}\rangle=|z%
\rangle$. In this case, there is a dramatic simplification for the shift:%
\begin{equation*}
\Delta\hat{a}=i\epsilon
O_{w}=\epsilon|O_{w}|e^{i(\theta_{z}+\theta_{w}-\pi/2)}
\end{equation*}
where $z=|z|e^{i\theta_{z}}$ and $O_{w}=|O_{w}|e^{i\theta_{w}}$. If we make
the following convenient choice for the phase $\theta_{z}=\pi/2$ and rewrite
the above equation in terms of the canonical pair $(\hat{Q},\hat{P})$, we
arrive at a \textit{symmetric} pair of equations for $\Delta\hat{Q}$ and $%
\Delta\hat{P}$:%
\begin{equation}
\Delta\hat{Q}=\epsilon\sqrt{2}|z|\operatorname{Re}(O_{w})\quad\text{and}\quad\Delta%
\hat{P}=\epsilon\sqrt{2}|z|\operatorname{Im}(O_{w})
\end{equation}

These equations do not depend on the quadratic dispersion or the time
derivative of the quadratic dispersion of any observable for the initial
state of the measuring system and, in principle, one may ``tune" the size of
the $\epsilon|z|$ term despite how small $\epsilon$ may be by making $|z|$
large enough. This is of great practical importance for optical
implementations of weak value measurements since $|z|$ for a quantized mode
of an electromagnetic field is nothing else but the mean photon number in
this mode for the coherent state $|z\rangle$ \cite{vedral}.

\section{von Neumann's Pre-measurement, Weak Values and the Geometry of
Quantum Mechanics}

In the last section we discussed von Neumann's model for a pre-measurement
in the weak measurement limit in order to obtain a deeper understanding of
the concept of a weak value in terms of a quantum phase space analysis of
the measuring apparatus system. In this section we implement, in a certain
sense, the opposite approach: We shall discuss certain geometric structures
of the measured system based on previous work of Tamate \textit{et al} \cite%
{Tamate}. Let $W=W_{(S)}\otimes W_{(M)}$ be the state space formed by
composing the subsystem $W_{(S)}$ with the measuring subsystem $W_{(M)}$. We
will initially assume that the measured system is a discrete quantum
variable of $W_{(S)}$ defined by an observable $\hat{O}=|o_{k}\rangle
o_{k}\langle o^{k}\vert$ (we use henceforth the sum convention). The
measuring subsystem will be considered as a structure-less (no spin or
internal variables) quantum mechanical particle in one dimension (further
ahead we will also consider discrete measuring systems). Suppose the initial
state of the total system is given by the following unentangled product
state: $|\Psi_{(I)}\rangle=|\alpha \rangle\otimes|\phi_{(I)}\rangle$. After
performing an \textit{ideal} von Neumann measurement through the interaction
Hamiltonian, 
\begin{equation}
\hat{H}_{int}=\lambda\delta(t-t_{0})\hat{O}\otimes\hat{P}
\label{ideal  von Neumann measurement interaction hamiltonian}
\end{equation}
the final state will be%
\begin{equation*}
(\hat{I}\otimes\langle
q(x)|)|\Psi_{(F)}\rangle=|o_{j}\rangle\alpha^{j}\phi_{(I)}(x-\lambda o_{j}),
\end{equation*}
where $\phi_{(I)}(x)=\left\langle q(x)\right\vert \phi_{(I)}\rangle$ is the
initial wave-function in the position basis of the measuring system. Note
that a correlation in the final state of the total system is then
established between the variable to be measured $o_{j}$ with the continuous
position variable of the measuring particle. This step of the von Neumann
measurement prescription is called the pre-measurement of the system.

Consider now the measuring system as a finite dimensional quantum system $%
W_{(M)}$. In particular, if $n=2$, our measuring apparatus consists of a
single qubit so that we can treat this two-level measuring system making
explicit use of the $\mathbb{CP}(1)$ (Bloch sphere) geometry. A single qubit
can be written in the Bloch sphere standard form as 
\begin{equation}
|\theta,\phi\rangle=\cos(\theta/2)|u_{0}\rangle+e^{i\phi}\sin(\theta
/2)|u_{0}\rangle  \label{bloch sphere parametrization}
\end{equation}

The single projective coordinate in this case is $\xi=\tan(\theta/
2)e^{i\phi }$ and, remarkably, we shall see that this complex number can
actually be measured physically as a certain appropriate weak value for two
level systems. Suppose now that the interaction happens in an arbitrary
finite dimensional measuring system: $W=W_{(S)}\otimes W_{(M)}$, that is, $%
\dim W_{(M)}=m$. The initial non-entangled pure-state is $%
|\Psi_{(I)}\rangle=|\alpha\rangle \otimes|\phi_{(I)}\rangle$ and the finite
momentum basis is given by $\{|v_{\sigma}\rangle\},(\sigma=0,1,\ldots,m-1)$
so that the momentum observable can be expressed as $\hat{P}%
=|v_{\sigma}\rangle p_{\sigma }\left\langle v^{\sigma}\right\vert $. Again
we model our instantaneous interaction with the interaction Hamiltonian
given by \eqref{ideal von Neumann measurement interaction hamiltonian} so
that our final entangled state is given by%
\begin{equation}
|\Psi_{(F)}\rangle=|A_{\sigma}\rangle\otimes|v_{\sigma}\rangle\phi^{\sigma},
\end{equation}
where 
\begin{equation}
|A_{\sigma}\rangle=e^{-i\lambda p_{\sigma}\hat{O}}|\alpha\rangle \qquad\text{%
and}\qquad|\phi_{(I)}\rangle=|v_{\sigma}\rangle\phi^{\sigma}
\end{equation}
The above entangled state clearly establishes a finite index correlation
between $|A_{\sigma}\rangle\in W_{(S)}$ and the finite momentum basis $%
|v_{\sigma}\rangle$. The total system is in the pure state $|\Psi
_{(F)}\rangle\langle\Psi_{(F)}\vert$ and by tracing out the first subsystem,
the measuring system will be:%
\begin{equation}
\hat{\rho}_{|\Psi_{(F)}\rangle}=|v_{\sigma}\rangle\phi^{\sigma}\left\langle
A^{\sigma}\right\vert A_{\tau}\rangle\phi_{\tau}\left\langle A^{\tau
}\right\vert
\end{equation}
For a single qubit, one has 
\begin{equation}
|\phi_{(I)}\rangle=\cos(\theta/2)|v_{0}\rangle+e^{i\phi}\sin(\theta
/2)|v_{1}\rangle,
\end{equation}
with 
\begin{equation}
\langle A^{0}\vert A_{1}\rangle=|\langle A^{0}\vert A_{1}\rangle|e^{-i\eta}
\end{equation}
so that we can compute the probability $P(\eta)$ of finding the second
subsystem in a certain reference state $|\theta=\pi/2,\varphi=0\rangle$ as%
\begin{align}
P(\eta)&=tr(\hat{\rho}_{|\Psi_{(F)}\rangle}|\pi/2,0\rangle\langle\pi
/2,0\vert)  \notag \\
&=\frac{1}{2}+\frac{1}{4}\left( |\langle A^{0}\vert
A_{1}\rangle|\sin\theta\cos(\phi-\eta)\right)
\end{align}
For a fixed $\theta$, this probability is clearly \textit{maximized} when $%
\mathit{\phi=\eta}$. This fact can be used to measure the so called \textit{%
geometric phase} 
\begin{equation*}
\eta=\arg\langle A^{0}\vert A_{1}\rangle
\end{equation*}
between the two indexed states $|A_{0}\rangle$ and $|A_{1}\rangle\in W_{(S)}$%
. This definition of a geometric phase was originally proposed in 1956 by
Pancharatnam \cite{pancharatnam} for optical states and rediscovered by
Berry in 1984 \cite{berry} in his study of the adiabatic cyclic evolution of
quantum states. In 1987, Anandan and Aharonov \cite{aharonovanandan1987}
gave a description of this phase in terms of natural geometric structures of
the $U(1)$ fiber-bundle structure over the space of rays and of the
symplectic and Riemannian structures in the projective space $\mathbb{CP}(N)$
inherited from the hermitian structure of $W_{(S)}$.

Given the final state $|\Psi_{(F)}\rangle$, one may then \textquotedblleft
post-select\textquotedblright\ a chosen state $|\beta\rangle\in W_{(S)}$.
The resulting state is then clearly%
\begin{equation}
|\Psi_{(F)}\rangle=C(\left\vert \beta\right\rangle \left\langle \beta
\right\vert \otimes\hat{I}))\left( \left\vert A_{\sigma}\right\rangle
\otimes\left\vert v_{\sigma}\right\rangle \phi^{\sigma}\right) ,
\end{equation}
where $C$ is an unimportant normalization factor. Because of the
post-selection, the system is now again in a non-entangled state so that the
partial trace of $|\Psi_{(F)}\rangle\langle\Psi_{(F)}|$ over the first
subsystem gives us 
\begin{equation}
|\phi_{(F)}\rangle=C\langle\beta|A^{0}\rangle\left\vert v_{\sigma
}\right\rangle \phi^{\sigma}
\end{equation}
By making the following phase choices 
\begin{equation}
\left\langle \beta\right\vert A^{0}\rangle=|\langle\beta|A^{0}\rangle
|e^{i\beta_{0}}\quad\text{and}\quad\left\langle \beta\right\vert
A^{1}\rangle=|\langle\beta|A^{1}\rangle|e^{-i\beta_{1}}
\end{equation}
we can again compute the probability of finding the second subsystem in
state $\left\vert \pi/2,0\right\rangle $ and again one finds that for a
fixed angle $\theta$, the maximum probability occurs for 
\begin{equation}
\phi\prime=\beta_{0}+\beta_{1}=\arg(\left\langle \beta\right\vert
A_{0}\rangle\left\langle A^{1}\right\vert \beta\rangle)
\end{equation}
This implies that there is an overall phase change $\Theta$ given by 
\begin{equation}
\Theta=\phi\prime-\phi=\arg(\left\langle A^{1}\right\vert \beta\rangle
\left\langle \beta\right\vert A^{0}\rangle\left\langle A^{0}\right\vert
A^{1}\rangle),
\end{equation}
which is a well-known geometric invariant in the sense that it depends only
on the projection of the state-vectors $\left\vert A_{0}\right\rangle $, $%
\left\vert A_{1}\right\rangle $ and $\left\vert \beta\right\rangle $ on the
Bloch sphere. In fact, this quantity is the intrinsic geometric phase picked
by a state-vector that is parallel transported through the closed geodesic
triangle defined by the projection of the three states on the space of rays.
For a single qubit, the geometric invariant is proportional to the area of
the geodesic triangle formed by the projection of the kets $\left\vert
A_{0}\right\rangle $, $\left\vert A_{1}\right\rangle $ and $\left\vert
\beta\right\rangle $ on the Bloch sphere and it is well known to be given by%
\begin{equation}
\Theta=\arg(\left\langle A^{1}\right\vert \beta\rangle\left\langle
\beta\right\vert A^{0}\rangle\left\langle A^{0}\right\vert A^{1}\rangle)=-%
\Omega
/2,
\end{equation}
where $\Omega$ is the oriented solid angle formed by the geodesic triangle. 
\begin{figure}[H]
\centering
\includegraphics[scale=.3]{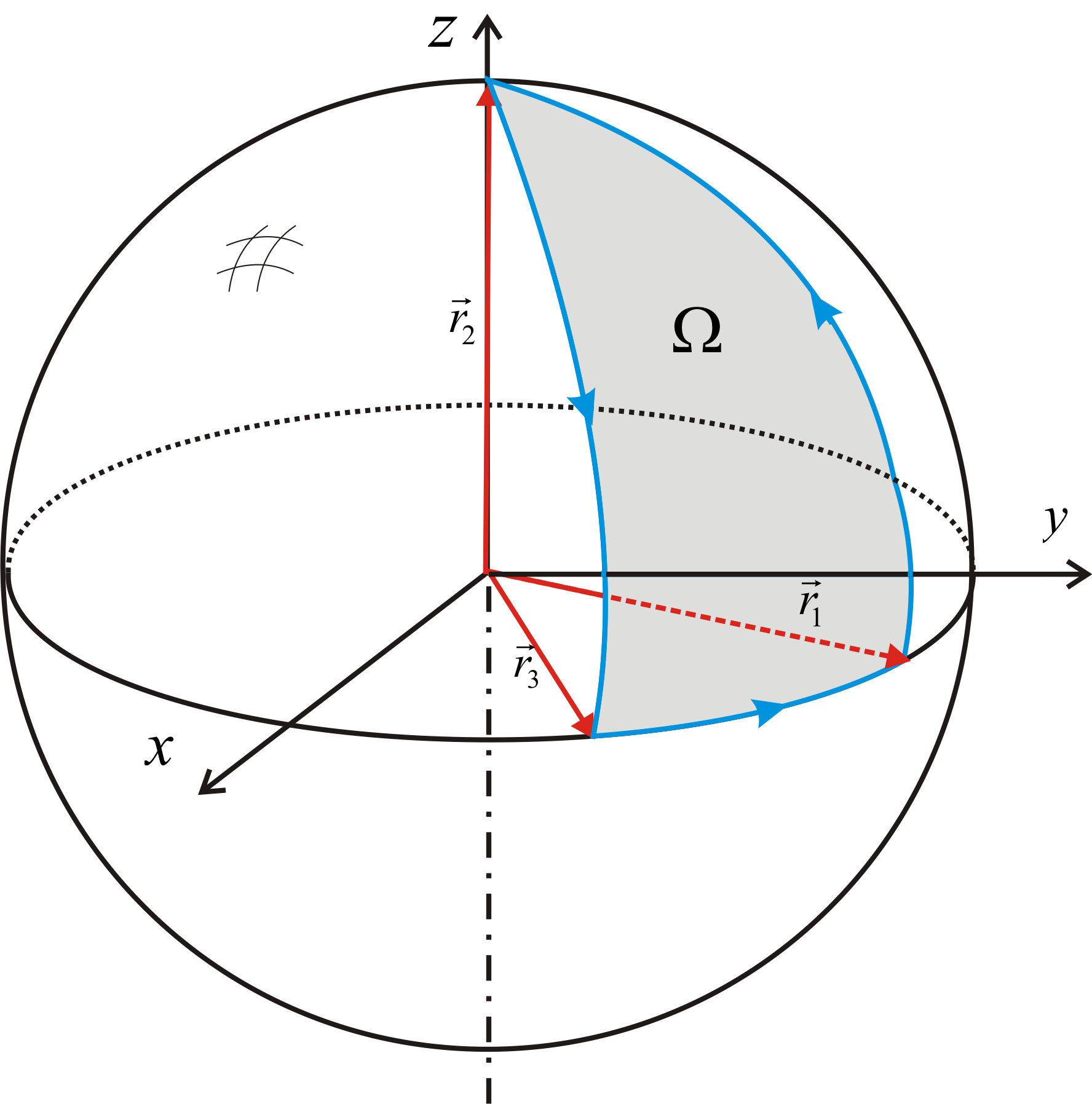}
\caption{{Solid angle formed by N. Pole and 2 Equator points.}}
\label{figsolidangle}
\end{figure}

Returning to the single qubit case, notice that if we choose the following
state $\left\vert \alpha\right\rangle =\left\vert u_{0}\right\rangle $ for
the pre-selected state (the \textquotedblleft north pole\textquotedblright\
of the Bloch sphere), and $|\theta,\phi\rangle$ for the post-selected state
and also 
\begin{equation}
\hat{O}=\hat{\sigma}_{1}=\left\vert u_{0}\right\rangle \left\langle
u^{1}\right\vert +\left\vert u_{1}\right\rangle \left\langle u^{0}\right\vert
\end{equation}
as the observable, then it is straightforward to compute the weak value as 
\begin{equation}
O_{w}=\tan(\theta/2)e^{i\phi},
\end{equation}
which is clearly complex-valued in general. What is curious about this
result is that the weak value gives a direct physical meaning to the complex
projective coordinate of the state vector in the Bloch sphere given by %
\eqref{projective coordinates} \cite{kobayashi2013}.
\begin{figure}[H]
\centering
\includegraphics[scale=.1]{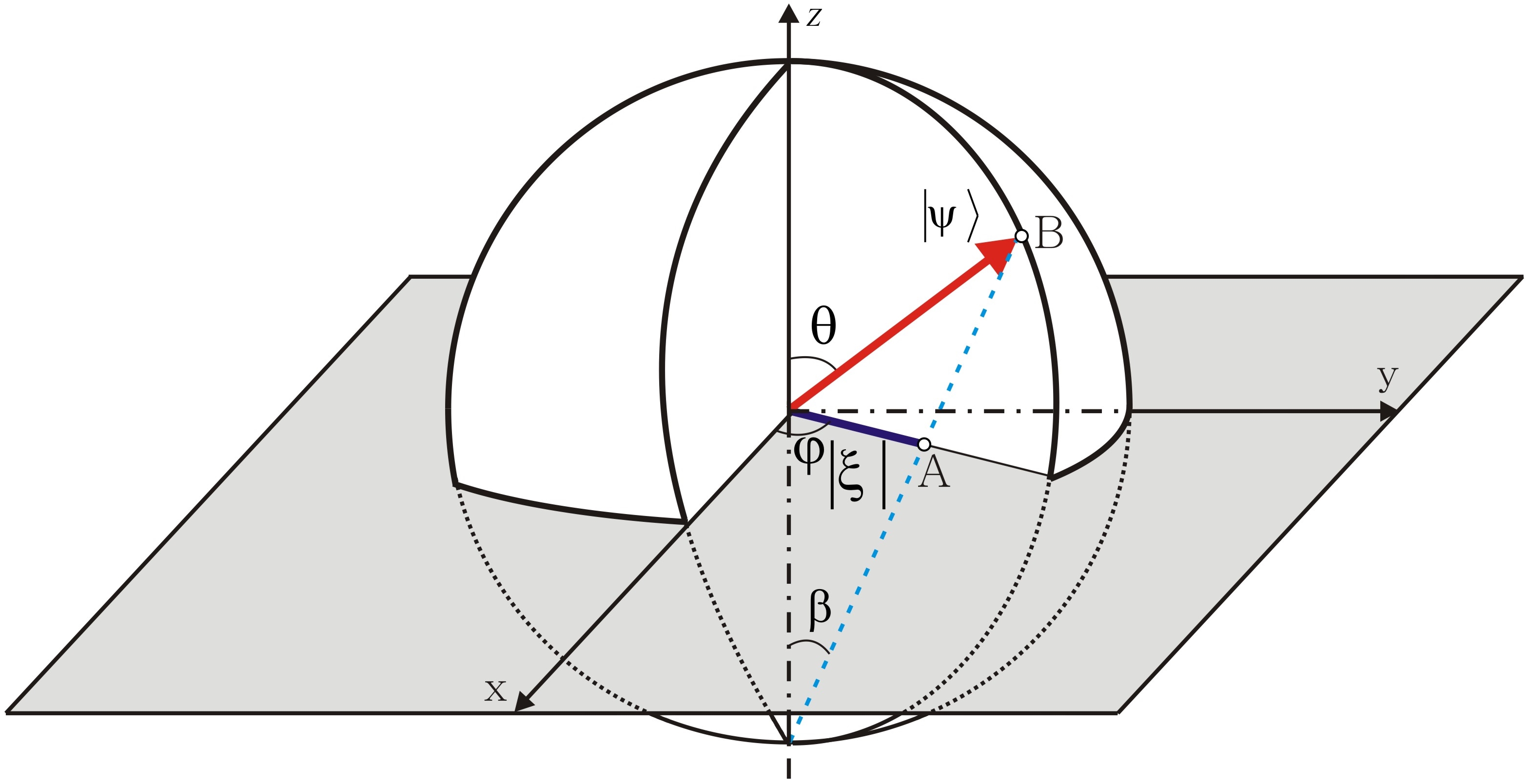}
\caption{{Stereographic projection.}}
\label{figstereographi}
\end{figure}

Suppose now that the physical system $W$ is composed by two subsystems $%
W_{(S)}\otimes W^{(\infty)}$ as before, but the measuring system $%
W^{(\infty)}$ is spanned by a complete basis of position kets $\{\left\vert
q(x)\right\rangle \}$ (momentum kets $\{\left\vert p(x)\right\rangle \}$),
with $-\infty<x,y<+\infty$. And again let us consider the initial state as
the product state vector $|\Psi_{(I)}\rangle=|\alpha\rangle\otimes|\phi
_{(I)}\rangle$ together with an instantaneous interaction coupling the
observable $\hat{O}$ of $W_{(S)}$ with the momentum observable $\hat{P}$ in $%
W^{(\infty)}$. The system evolves then to%
\begin{equation}
|\Psi_{(F)}\rangle=\int_{-\infty}^{+\infty}dy\left\vert A(y)\right\rangle
\otimes\left\vert p(y)\right\rangle \widetilde{\phi}_{(I)}(y),
\end{equation}
where%
\begin{equation}
\left\vert A(y)\right\rangle =e^{-i\lambda y\hat{O}}\left\vert \alpha
\right\rangle
\end{equation}
is the continuous indexed states that are correlated to the momentum basis
of the measuring apparatus and 
\begin{equation}
\widetilde{\phi}_{(I)}(y)=\left\langle p(y)\right\vert \phi_{(I)}\rangle
\end{equation}
is the wave function of the initial state of the apparatus in the \textit{%
momentum basis}. We may now compute (to first order in $dy$) the intrinsic
phase shift between $\left\vert A(y)\right\rangle $ and $\left\vert
A(y+dy)\right\rangle $ in a similar manner that was carried out before with
the discretely parameterized states:%
\begin{equation}
\arg(\left\langle A(y)\right\vert A(y+dy)\rangle)\approx-\lambda dy\langle 
\hat{O}\rangle_{\left\vert \alpha\right\rangle },
\label{argument of inner product of nearby states}
\end{equation}
where $\langle\hat{O}\rangle_{\left\vert \alpha\right\rangle }$ is the
expectation value 
\begin{equation}
\langle\hat{O}\rangle_{\left\vert \alpha\right\rangle }=\left\langle
\alpha\right\vert \hat{O}\left\vert \alpha\right\rangle
\end{equation}
is the expectation value of observable $\hat{O}$ in the initial state $%
\left\vert \alpha\right\rangle $. We can also compute the shift of the
expectation value of the position observable of the particle of the
measuring system between the initial and final states. Let $\{\left\vert
o_{j}\right\rangle \},\ (j=0,1,\ldots N-1)$ be a complete set of eigenkets
of observable $\hat{O}$. The final state of the composite system can be
described by the following pure density matrix:%
\begin{align}
\hat{\rho}_{|\Psi(F)\rangle}&=|\Psi_{(F)}\rangle\langle\Psi_{(F)}|  \notag \\
&=|o_{j}\rangle\langle o^{k}|\otimes\alpha^{j}\hat{V}_{\lambda o_{j}}^{\dag
}|\phi_{(I)}\rangle\langle\phi_{(I)}|\hat{V}_{\lambda o_{k}}\alpha_{k}
\end{align}
By taking the partial trace of the $W_{(S)}$ system, we arrive at the
following mixed state that describes the measuring system at instant $t_{F}$:%
\begin{equation}
\hat{\rho}_{\left\vert \Psi(F)\right\rangle }^{(M)}=\underset{j}{\sum}%
|\alpha_{j}|^{2}\hat{V}_{\lambda o_{j}}^{\dag}\left\vert
\phi_{(I)}\right\rangle \left\langle \phi_{(I)}\right\vert \hat{V}_{\lambda
o_{j}}
\end{equation}
The \textit{ensemble} expectation value of position is then%
\begin{equation*}
\lbrack\hat{Q}]_{\hat{\rho}_{\left\vert \Psi(F)\right\rangle }^{(M)}}=tr(%
\hat{\rho}_{\left\vert \Psi(F)\right\rangle }^{(M)}\hat{Q})=\langle\hat {Q}%
\rangle_{|\phi_{(I)}\rangle}+\lambda\langle\hat{O}\rangle_{\left\vert
\alpha\right\rangle }
\end{equation*}

A geometric interpretation of this von Neumann's pre-measurement can be
presented in the following way: Let $\left\vert \psi(t)\right\rangle $ be
the curve of normalized state vectors in $W^{(N+1)}$ given by the unitary
evolution generated by a given Hamiltonian $\hat{H}(t)$. The Schr\"{o}dinger
equation implies a relation between $\left\vert \psi(t)\right\rangle $ and $%
\left\vert \psi(t+dt)\right\rangle $ given by:%
\begin{equation*}
\left\vert d\psi(t)\right\rangle =\left\vert \psi(t+dt)\right\rangle
-\left\vert \psi(t)\right\rangle =-i\hat{H}\left\vert \psi(t)\right\rangle dt
\end{equation*}
The above equation together with \eqref{distance element in CP(N)} lead to a
very elegant relation for the squared distance between two infinitesimally
nearby projection of state vectors connected by the unitary evolution over $%
\mathbb{CP}(N)$ \cite{aharonovanandan1987}:%
\begin{align}
ds^{2}(\mathbb{CP}(N))&=[\left\langle \psi(t)\right\vert \hat{H}%
^{2}\left\vert \psi(t)\right\rangle -\left\langle \psi(t)\right\vert \hat{H}%
\left\vert \psi(t)\right\rangle ^{2}]dt^{2}  \notag \\
&=\delta_{\left\vert \psi(t)\right\rangle }^{2}E
\label{dinamica de uma variacao infinitesimal sobre CP(N)}
\end{align}
The above equation means that the speed of the projection over $\mathbb{CP}%
(N)$ equals the instantaneous energy uncertainty 
\begin{equation}
\frac{ds}{dt}=\delta_{\left\vert \psi(t)\right\rangle }E
\end{equation}

A beautiful geometric derivation of the time-energy uncertainty relation
that follows directly from this equation can be found in \cite%
{aharonovanandan1987}. (See also \cite{loboribeiroribeiro} for a pedagogical
discussion of this result related to the adiabatic theorem in quantum
mechanics).

Back to our discussion of the interaction between the systems $W^{(S)}$ and $%
W^{(\infty)}$, note that the expression $\left\vert A(y)\right\rangle
=e^{-i\lambda y\hat{O}}\left\vert \alpha\right\rangle $ is formally
equivalent to the unitary time evolution equation $\left\vert
\psi(t)\right\rangle =e^{-i\hat{H}t}\left\vert \psi(0)\right\rangle $ which
is clearly a solution of the Schr\"{o}dinger equation with a
time-independent Hamiltonian. A formal analogy between the two distinct
physical processes is exemplified by the association below: 
\begin{equation*}
\begin{array}{ll}
|\psi(t)\rangle & \mapsto|A(y)\rangle \\ 
|\psi(0)\rangle & \mapsto|\alpha\rangle=|A(0)\rangle \\ 
t & \mapsto y \\ 
\hat{H} & \mapsto\lambda\hat{O}.%
\end{array}%
\end{equation*}
Looking at subsystem $W^{(S)}$ and regarding $y$ as an external parameter
(just like the time variable for the unitary time evolution) we may write
the analog of \eqref{dinamica de uma variacao infinitesimal sobre CP(N)} in $%
\mathbb{CP}(N)\subset W^{(S)}$:%
\begin{align}
ds^{2}&=[\left\langle A(y)\right\vert \hat{O}^{2}\left\vert
A(t)\right\rangle -\left\langle A(y)\right\vert \hat{O}\left\vert
A(t)\right\rangle ^{2}]\lambda^{2}dy^{2}  \notag \\
&=[\left\langle \alpha\right\vert \hat{O}^{2}\left\vert \alpha\right\rangle
-\left\langle \alpha\right\vert \hat{O}\left\vert \alpha\right\rangle
^{2}]\lambda^{2}dy^{2}
\end{align}
Comparing this result with equations \eqref{distance element in CP(N)} and %
\eqref{argument of inner product of nearby states} we can immediately see
the geometric interpretation for the expectation value $\left\langle
\alpha\right\vert \hat{O}\left\vert \alpha\right\rangle $ in terms of the $%
U(1)$ fiber bundle structure as one can easily infer from the pictorial
representation in Figure \ref{figphasedifference}.
\begin{figure}[H]%
\centering
\includegraphics[
natheight=1.760800in,
natwidth=2.187100in,
height=1.7979in,
width=2.2278in
]%
{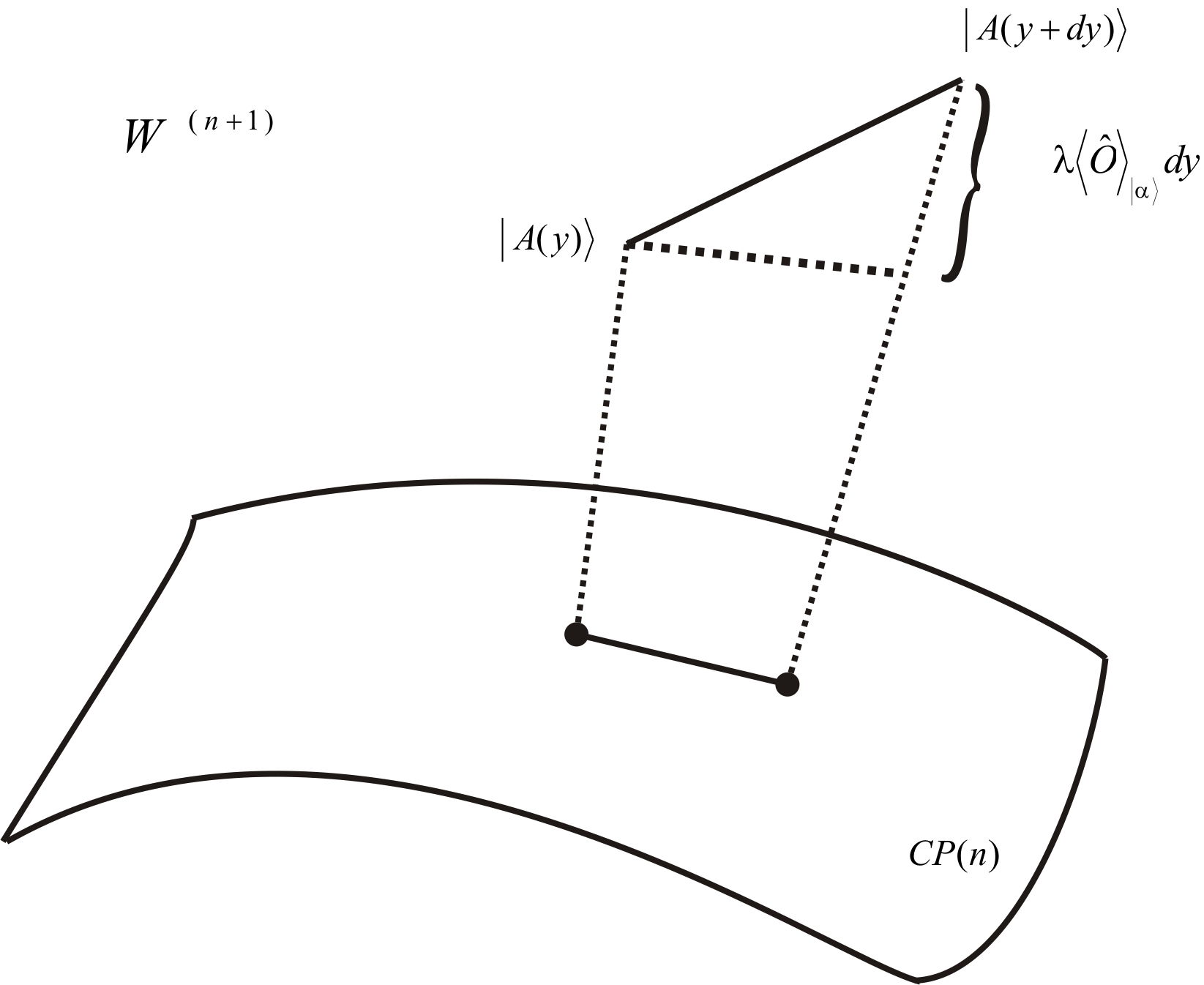}
\caption{{Phase difference between  $|A(y)\rangle$\ and \ $|A(y+dy)\rangle$.}}
\label{figphasedifference}
\end{figure}

For the case of a weak measurement, we choose again the Hamiltonian given by
equation \eqref{hamiltonian for weak measurement}. Given the initial
unentangled state $|\Psi _{(I)}\rangle =|\alpha \rangle \otimes |\phi
_{(I)}\rangle $ at $t_{0}$, the evolution of the system is described as%
\begin{align}
& |\Psi _{(F)}\rangle =\int_{-\infty }^{+\infty }dy\left\vert
A(y)\right\rangle \otimes \left\vert p(y)\right\rangle \widetilde{\phi }%
_{(I)}(y)\text{, with}  \notag \\
& \left\vert A(y)\right\rangle =e^{-i\lambda y\hat{O}}\left\vert \alpha
\right\rangle  \notag
\end{align}%
The global geometric phase related to the infinitesimal geodesic triangle
formed by the projections of $\left\vert A(y)\right\rangle $, $\left\vert
A(y+dy)\right\rangle $ and the post-selected state $\left\vert \beta
\right\rangle $ on $\mathbb{CP}(N)$ (see Figure \ref{figglobalgeomphase}) is
clearly given by%
\begin{equation}
\Theta =\arg (\left\langle A(y)\right\vert \beta \rangle \left\langle \beta
\right\vert A(y+dy)\rangle \left\langle A(y+dy)\right\vert A(y)\rangle )
\end{equation}%
\begin{figure}[H]
\centering
\includegraphics[
natheight=1.760800in,
natwidth=2.312500in,
height=1.7979in,
width=2.3531in
]{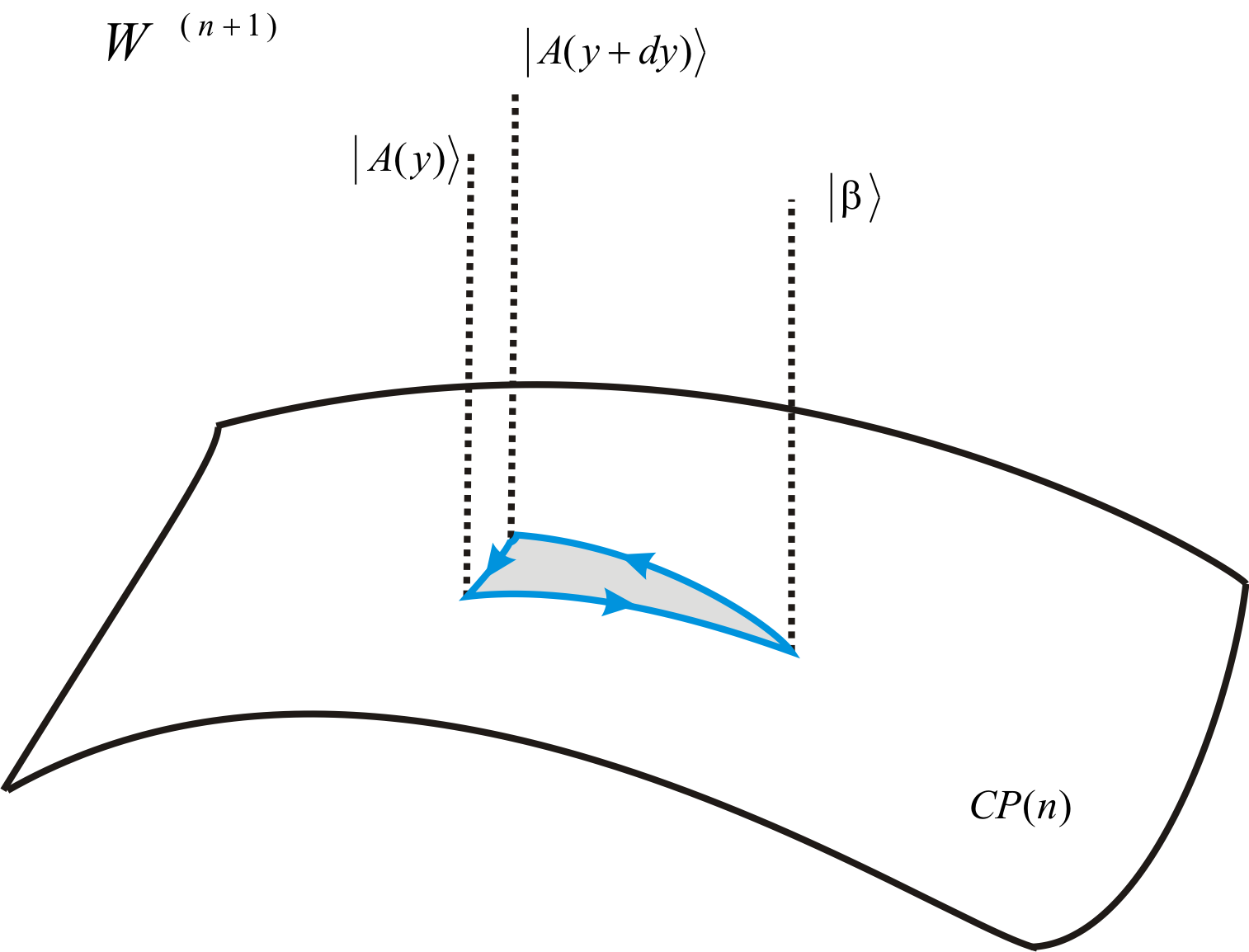}
\caption{{Representation of global geometric phase.}}
\label{figglobalgeomphase}
\end{figure}

\section{The Geometry of Deterministic Measurements}

As is well-known, the collapse postulate of Quantum Mechanics implies that,
in general, the measurement of a quantum state causes a disturbance of the
state as the state ``jumps" to an eigenstate of the observable that is being
measured in a stochastic manner. Yet, if the state is \textit{already} an
eigenstate of the observable, then the state is left untouched. This is what
is behind the notion of ``deterministic" experiments. We call a measurement,
a deterministic experiment when we measure only variables for which the
state of the system under investigation is an eigenstate. In other words,
for any state $\left\vert \psi\right\rangle $ it is possible to ask the
following question: ``What is the set of Hermitian operators $A_{\left\vert
\psi\right\rangle }$ for which $\left\vert \psi\right\rangle $ is an
eigenstate?" That is, which satisfy: 
\begin{equation}
A_{\left\vert \psi\right\rangle }=\{\hat{A}_{i}\,\,\text{such\thinspace
\thinspace that}\,\,\hat{A}_{i}|\psi\rangle=a_{i}|\psi\rangle\,,a_{i}\in%
\mathbf{%
\mathbb{R}
}\}
\end{equation}
One can think this as kind of a dual question to the more familiar inquire
which asks ``What are the eigenstates of a given operator?" The measurements
of such operators $A_{\left\vert \psi\right\rangle }$ would not lead to any
wave function collapse, since the wave function is initially an eigenstate
of the operator being measured. The results are completely predictable and
so the experiments are \textit{deterministic} in this sense. Given $%
\left\vert \psi\right\rangle $, one can characterize mathematically in a
very precise way the \textit{deterministic set of operators} (DSO) $%
A_{\left\vert \psi\right\rangle }$. In fact, let $\hat{A}_{i},\hat{A}_{j}\in
A_{\left\vert \psi\right\rangle }$, then the set $A_{\left\vert \psi
\right\rangle }$ is \textit{closed} in the sense that:

\begin{enumerate}
\item $[\hat{A}_{i},\hat{A}_{j}]\in A_{\left\vert \psi\right\rangle }$

\item $\alpha\hat{A}_{i}+\beta\hat{A}_{j}\in A_{\left\vert \psi\right\rangle
}$ with $\alpha,\beta\in%
\mathbb{R}
$

\item $\hat{A}_{i}\hat{A}_{j}$ $\in A_{\left\vert \psi\right\rangle }$
\end{enumerate}

Note that $[\hat{A}_{i},\hat{A}_{j}]$ is a special deterministic operator in
the sense that it annihilates $\left\vert \psi\right\rangle $, $[\hat{A}_{i},%
\hat{A}_{j}]\left\vert \psi\right\rangle =0$. This implies that $[\hat {A}%
_{i},\hat{A}_{j}]$ must be proportional to the projector $\hat{\Pi }%
_{\left\vert \psi\right\rangle }^{\perp}=\hat{I}-\hat{\Pi}_{\left\vert
\psi\right\rangle }=\hat{I}-\left\vert \psi\right\rangle \left\langle
\psi\right\vert $ that projects vectors to the subspace that is \textit{%
orthogonal} to $\left\vert \psi\right\rangle $.

Note that the projector operators $\hat{\Pi}_{\left\vert \psi\right\rangle }$
and $\hat{\Pi}_{\left\vert \psi\right\rangle }^{\perp}$ are clearly \textit{%
idempotent}, that is $\hat{\Pi}_{\left\vert \psi\right\rangle }^{2}=\hat{\Pi}%
_{\left\vert \psi\right\rangle }$ and $(\hat{\Pi }_{\left\vert
\psi\right\rangle }^{\perp})^{2}=\hat{\Pi}_{\left\vert
\psi\right\rangle }^{\perp}$.

Also, if we know the Hamiltonian $\hat{H}$ of the system, then we know the
set $A_{\left\vert \psi(t)\right\rangle }$ for each instant of time, where $%
|\psi(t)\rangle$ is a solution of Schr\"{o}dinger's equation $\hat{H}%
|\psi(t)\rangle=i|\dot{\psi}(t)\rangle$.

If we have a Hilbert space of dimension $n$ then we can choose an
orthonormal basis such that the state vector $|\psi\rangle$ can be
represented by the $n\times1$ column vector below (a unitary transformation
can always bring us to such a basis) together with any orthogonal vector $%
|\psi_{\bot}\rangle$%
\begin{equation}
|\psi\rangle=\left( 
\begin{array}{ccccc}
1 & 0 & 0 & \ldots & 0%
\end{array}
\right)^{T}\text{, and\quad}|\psi_{\bot}\rangle=\left( 
\begin{array}{ccccc}
0 & a^{1} & a^{2} & \ldots & a^{n}%
\end{array}
\right) ^{T}  \label{basis for psi}
\end{equation}

The Hermitian operators that operate on a $n$-di-\ mensional space form an 
\textit{real} $n^{2}$-dimensional algebra. The DSO that act upon this space
(in the above basis) can be represented by

\begin{equation}
\left( 
\begin{array}{ccccc}
a_{1}^{1} & 0 & 0 & \ldots & 0 \\ 
0 & a_{2}^{2} & a_{3}^{2} & \ldots & a_{n}^{2} \\ 
0 & a_{2}^{3} & a_{3}^{3} & \ldots & a_{3}^{3} \\ 
\vdots & \vdots & \vdots & \ddots & \vdots \\ 
0 & a_{2}^{n} & a_{3}^{n} & \ldots & a_{n}^{n}%
\end{array}
\right)  \label{matrix form for DSO}
\end{equation}
with $a_{j}^{i}=\bar{a}_{i}^{j}$, so that the dimension of the DSO space is
clearly $(n-1)^{2}+1$.

\subsection{The Completely Uncertain Operators}

In addition to the set $A_{\left\vert \psi\right\rangle }$ of DSO, there is
also a set $B_{\left\vert \psi\right\rangle }$ of operators whose results
are \textit{completely} \textit{uncertain}. In fact, given $\left\vert
\psi\right\rangle $, we can always decompose any operator $\hat{C}$ in the
following unique way: 
\begin{equation}
\hat{C}\equiv\hat{A}_{\left\vert \psi\right\rangle }^{(C)}+\hat{B}%
_{\left\vert \psi\right\rangle }^{(C)}
\end{equation}
where 
\begin{align}
&\hat{A}_{\left\vert \psi\right\rangle }^{(C)}=\hat{\Pi}_{\left\vert
\psi\right\rangle }\hat{C}\hat{\Pi}_{\left\vert \psi\right\rangle }+\hat{\Pi 
}_{\left\vert \psi\right\rangle }^{\perp}\hat{C}\hat{\Pi}_{\left\vert
\psi\right\rangle }^{\perp}\text{, and}  \notag \\
&\hat{B}_{\psi}^{(C)}=\hat{\Pi}_{\left\vert \psi\right\rangle }^{\perp}\hat{C%
}\hat{\Pi}_{\left\vert \psi\right\rangle }+\hat{\Pi}_{\left\vert
\psi\right\rangle }\hat{C}\hat{\Pi }_{\left\vert \psi\right\rangle }^{\perp}
\end{align}

In the same basis of \eqref{matrix form for DSO}, the \textit{completely
uncertain operators} (CUO) have the following matrix form: 
\begin{equation}
\hat{B}\left\vert \psi\right\rangle =\left( 
\begin{array}{ccccc}
0 & b_{2}^{1} & b_{3}^{1} & \ldots & b_{n}^{1} \\ 
b_{1}^{2} & 0 & 0 & \ldots & 0 \\ 
b_{1}^{3} & 0 & 0 & \ldots & 0 \\ 
\vdots & \vdots & \vdots & \ddots & \vdots \\ 
b_{n}^{3} & 0 & 0 & \ldots & 0%
\end{array}
\right) \left( 
\begin{array}{c}
1 \\ 
0 \\ 
0 \\ 
\vdots \\ 
0%
\end{array}
\right) =\left( 
\begin{array}{c}
0 \\ 
b_{1}^{2} \\ 
b_{1}^{3} \\ 
\vdots \\ 
b_{n}^{3}%
\end{array}
\right) =\lambda\left\vert \psi_{\bot}\right\rangle
\label{matrix form for  CUO}
\end{equation}
Note that the uncertain operator $\hat{B}$ takes the state $\left\vert
\psi\right\rangle $ to an orthogonal normalized state $|\psi_{\bot}\rangle$
(up to a normalizing factor $\lambda$) and that the dimension of the CUO
subspace of operator subspace is clearly $2(n-1)$. In this way, we see that
the $n^{2}$-dimensional space of Hermitian operators is decomposed into the 
\textit{direct orthogonal sum} of the DSO and CUO subspaces: 
\begin{equation}
\dim(DSO\oplus CUO)=n^{2}=\dim(DSO)+\dim(CUO)
\end{equation}%
\begin{figure}[H]%
\centering
\includegraphics[scale=.3]{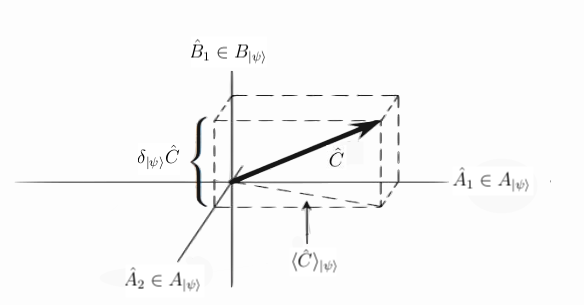}%
\caption{Pictorial representation in Operator Space.}%
\end{figure}
Traditionally, it was believed that if a measurement interaction is weakened
so that there is no disturbance on the system, then no information could be
obtained. However, the advent of the concept of weak measurements and weak
values has changed the point of view on this issue quite dramatically. Not
only we have learned that there is indeed a gain of information, but it
turned-out to be quite a very important tool both theoretically and for
practical purposes \cite{Yutaka2012,dressel2014}.

It has been shown \cite{spie-nswm} that information can be obtained even
though not a single particle (in an ensemble) is disturbed. To begin to
introduce this point, let us consider a general theorem for any vector in
Hilbert space:

\begin{theorem}
Let $W$ be a Hilbert space and an arbitrary non-null normalized vector
$|\psi\rangle\in W$ and also let $\hat{A}$ be an arbitrary Hermitian operator
that acts upon $W$, then:%
\begin{equation}
\hat{A}|\psi\rangle=\langle\hat{A}\rangle_{|\psi\rangle}|\psi\rangle
+\delta_{|\psi\rangle}\hat{A}|\psi_{\perp}\rangle
\end{equation}
where $\langle\hat{A}\rangle_{|\psi\rangle}=\left\langle \psi\right\vert
\hat{A}|\psi\rangle$ is the expectation value of observable $\hat{A}$ in state
$|\psi\rangle$, $\delta_{|\psi\rangle}^{2}\hat{A}=\left\langle \psi\right\vert
\hat{A}^{2}|\psi\rangle-\left\langle \psi\right\vert \hat{A}|\psi\rangle^{2}$
is the squared uncertainty of observable $\hat{A}$ in state $|\psi\rangle$, and
$|\psi_{\perp}\rangle$ is a vector that belongs to the subspace of $W$ that is
orthogonal to $|\psi\rangle$.
\end{theorem}

\begin{proof}
Left multiplication by $\langle\psi|$ yields the first term. By using that
$\hat{\Pi}_{\left\vert \psi\right\rangle }^{\perp}=|\psi_{\perp}\rangle
\langle\psi_{\perp}|=\hat{I}-\hat{\Pi}_{\left\vert \psi\right\rangle }=\hat
{I}-|\psi\rangle\langle\psi|$ and also by evaluating $|(\hat{A}-\langle\hat
{A}\rangle_{|\psi\rangle}\hat{I})|\psi\rangle|^{2}=\delta_{|\psi\rangle}%
^{2}\hat{A}$ yields the second term.
\end{proof}
\begin{figure}[H]
\centering
\includegraphics[scale=.3]{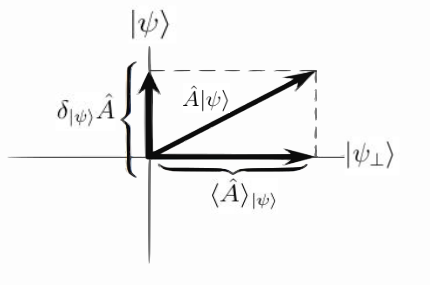}
\caption{Pictorial Representation in Hilbert Space.}
\end{figure}

Using this theorem, it is easy to see that in the basis of eq. 
\eqref{basis
for psi}, the diagonal elements of \eqref{matrix form for DSO} correspond to
the averages of the observables, e.g. $a_{1}^{1}=\langle\hat {A}%
\rangle_{\left\vert \psi\right\rangle }$ for state $\left\vert
\psi\right\rangle $, which again can be measured without uncertainty.

By normalizing the vector in the second member of eq. 
\eqref{matrix form for
CUO}, one has for its square modulus 
\begin{equation}
\lambda^{2}=|b_{1}^{2}|^{2}+|b_{1}^{3}|^{2}+\cdot\cdot+|b_{1}^{n}|^{2}
\end{equation}

From the above theorem, it is also easy to see that the normalizing factor $%
\lambda$ in eq. \eqref{matrix form for CUO} is nothing else but the
uncertainty of $\hat{B}$: 
\begin{equation}
\lambda=\delta_{\left\vert \psi\right\rangle }\hat{B}
\end{equation}

\subsection{Deterministic and Partially Deterministic Protocols}

\subsubsection{The Case of a Single qubit}

Consider a quantum system represented by a single qubit as in the spin of a
electron for instance. The best way to ``visualize" this system is to
represent its states on the Bloch sphere. Suppose that the north and south
``poles" of the sphere are represented by the orthonormal basis $\left\{
\left\vert u_{0}\right\rangle ,\left\vert u_{1}\right\rangle \right\} $ and
let 
\begin{equation}
\left\vert \psi_{\alpha}\right\rangle =\frac{1}{\sqrt{2}}\left( \left\vert
u_{0}\right\rangle +e^{i\alpha}\left\vert u_{1}\right\rangle \right)
\label{arbitrary state on the equator}
\end{equation}
be an arbitrary state represented on the ``equator" of the sphere as in
shown below: 
\begin{figure}[H]
\begin{center}
\includegraphics[scale=.25]{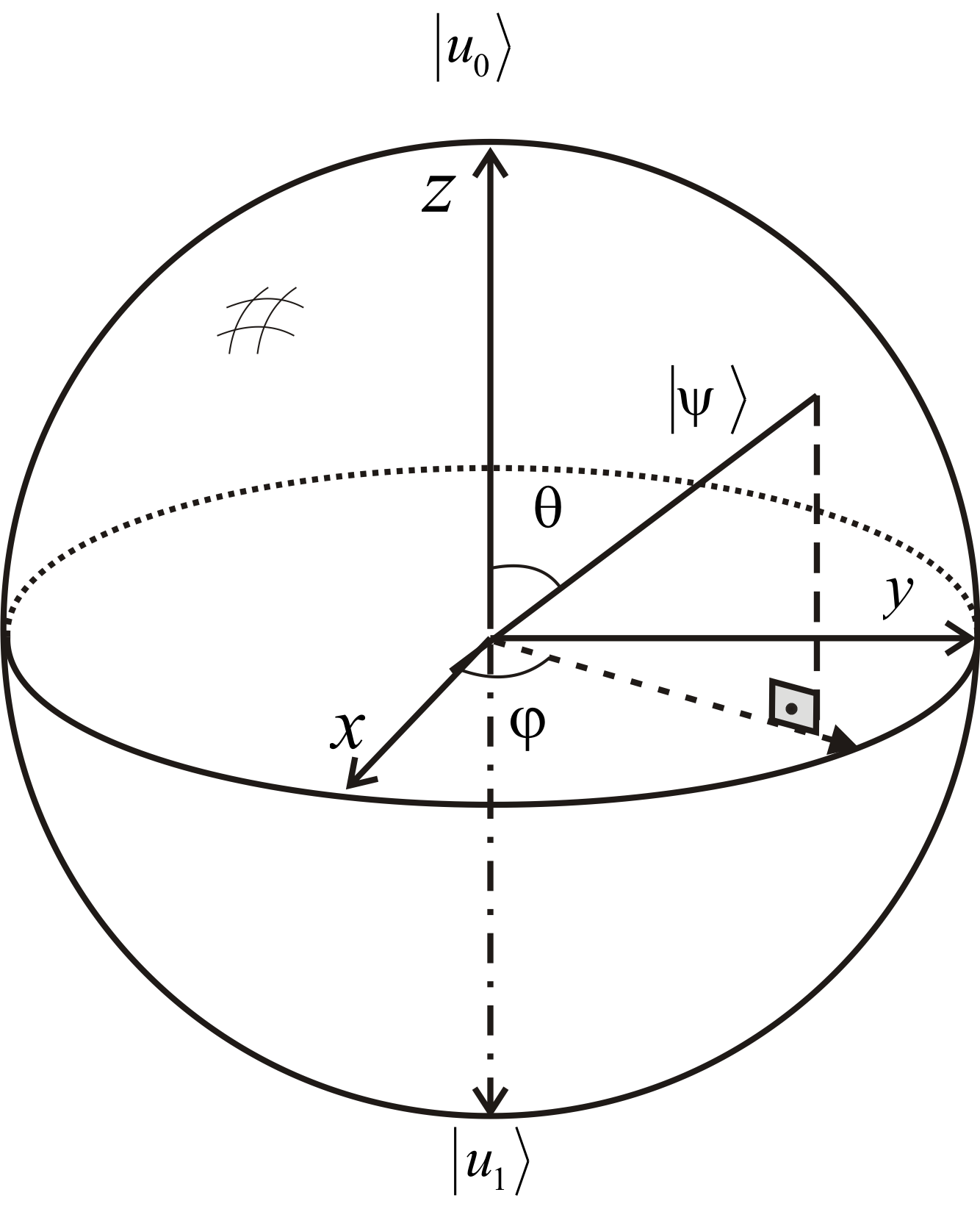}
\end{center}
\caption{Bloch sphere}
\end{figure}

Let us introduce now a game between two parties, Alice and Bob. Alice
chooses one of the two phases: $\alpha=0$ or $\alpha=\pi$ so that she has
the following two states at her disposal: 
\begin{equation}
\left\vert \pm\right\rangle =\frac{1}{\sqrt{2}}\left( \left\vert
u_{0}\right\rangle \pm\left\vert u_{1}\right\rangle \right)
\label{x+ and x- states}
\end{equation}
She delivers one of them to Bob and his task is to discover which one. All
he has to do is to measure the $\hat{\sigma}_{1}$ observable: 
\begin{equation}
\langle\hat{\sigma}_{1}\rangle_{\left\vert \pm\right\rangle }=\left\langle
\pm\right\vert \hat{\sigma}_{1}\left\vert \pm\right\rangle =\pm1
\end{equation}
Thus, this is an example of a deterministic measurement because Bob has the
previous knowledge that the state is an eigenstate of $\hat{\sigma}_{1}$ and
so with one single measurement he discovers what was Alice's choice and one
classical bit has been transmitted between them. Suppose now that Alice
chooses an arbitrary phase $\alpha$ and Bob must find out \textit{what} is
its value. In this case, Bob must measure a great number of equally prepared 
$\left\vert \psi_{\alpha}\right\rangle $ states so that he can find the
average 
\begin{equation}
\langle\hat{\sigma}_{1}\rangle_{\left\vert \psi_{\alpha}\right\rangle
}=\left\langle \psi_{\alpha}\right\vert \hat{\sigma}_{1}\left\vert
\psi_{\alpha}\right\rangle =\cos\alpha
\end{equation}
After a sufficiently large number of measurements, he can find $\cos\alpha$.
Since $\cos\alpha=\cos\left( 2\pi-\alpha\right) $ there is still an
ambiguity between these two possible states, but a few more measurements in
the direction $\vec{\sigma}\cdot\hat{n}_{\alpha}$ where $\hat{n}%
_{\alpha}=\cos\alpha\hat{\imath}+\sin\alpha\hat{\jmath}$ solves the problem
with ease. Thus, Bob can determine the phase $\alpha$ with arbitrary
precision given a sufficiently large ensemble of equally prepared states $%
\left\vert \psi_{\alpha}\right\rangle $.

\subsection{The Measurement of Interference for a Single Particle}

Consider the quantum mechanic unidimensional motion of a \textit{single}
particle. Suppose that the particle in the instant $t=0$ is described as a
sum of two highly concentrated ``lumps" of the wave function at two
macroscopically separated points. For instance, consider that at $t=0$, the
single particle has just emerged from the interaction with a two-slit
apparatus (with length $L$ between the slits). It is reasonable to expect
that at this particular moment, the wave-function in the \textit{position}
basis is indeed highly localized around each one of the two slits. Let us
consider the degree of freedom along the apparatus as the $y$ direction and
the direction of the incident particle beam as the $x$-direction as shown in
the figure \ref{twoslitfigure}. 
\begin{figure}[H]
\begin{center}
\includegraphics[scale=.25]{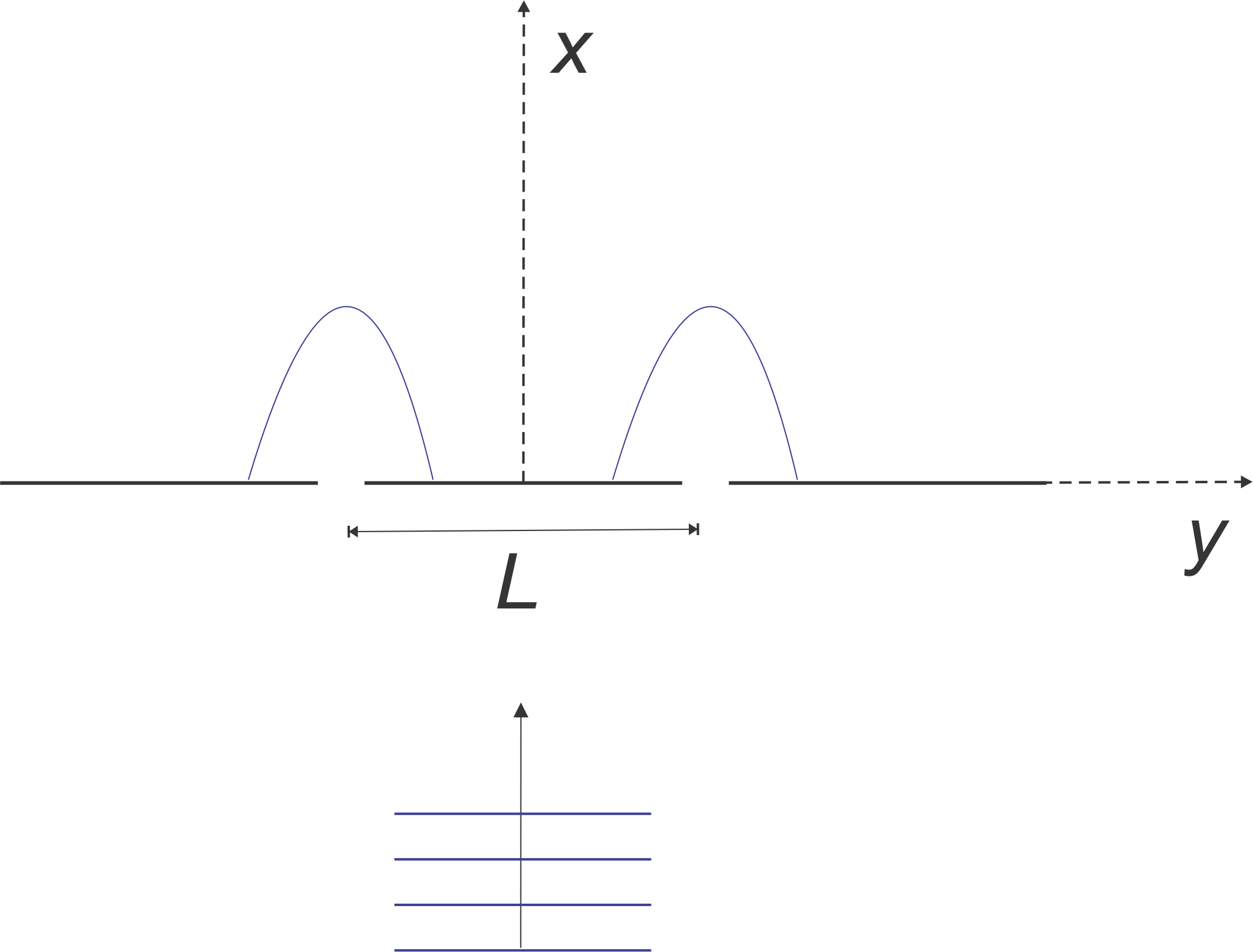}
\end{center}
\caption{Two-Slit}
\label{twoslitfigure}
\end{figure}

Let the particle's wave-function at $t=0$ be given by the following Schr\"{o}%
dinger-Cat like state as a linear combination of two macroscopically
separated coherent states (\ref{def. of coherent states})%
\begin{align}
&\left\vert \varphi_{1}(0)\right\rangle =|p=0,q=-L/2\rangle\text{, and} 
\notag \\
&\left\vert \varphi_{2}(0)\right\rangle =|p=0,q=+L/2\rangle
\end{align}
where 
\begin{equation}
|p,q\rangle=\hat{D}[p,q]|0\rangle=e^{-iqp/2}\hat{U}_{p}\hat{V}_{q}^{\dag
}|0\rangle  \label{definition of a coherent state 2}
\end{equation}
We are additionally supposing that the distance $L$ is much larger than the
spread of the Gaussian wave functions (which in this case is given by $%
\delta=1/\sqrt{2}$) so that the overlap between $\left\vert \varphi
_{1}(0)\right\rangle $ and $\left\vert \varphi_{2}(0)\right\rangle $ is 
\textit{negligible}. Indeed, because of (\ref{overlap of coherent states}),
we have that $\left\vert \left\langle \varphi_{2}(0)\right\vert \varphi
_{1}(0)\rangle\right\vert ^{2}=e^{-L^{2}/2}$. Suppose now that we take $L=10$%
, then \newline
$|\langle\varphi_{2}(0)|\varphi_{1}(0)\rangle|^{2}\simeq2\times10^{-22}$
which implies that for all practical purposes the two states can indeed be
considered \textit{orthogonal} to each other. A general global normalized
state of the particle can then approximately be written as the following
linear combination 
\begin{equation}
\left\vert \varphi_{\alpha}(0)\right\rangle \approx\frac{1}{\sqrt{2}}\left(
|0,-L/2\rangle+e^{i\alpha}|0,+L/2\rangle\right)  \label{arbitrary cat state}
\end{equation}
Note that the above equation actually describes a \textit{family} of states
parametrized by the relative phase $\alpha$. The question here is the same
that we asked in the previous subsection for the electron's spin: How can we
measure the relative phase? Note that no \textit{local} measurements of the
separate wave functions 
\begin{align}
&\langle q(y)\left\vert \varphi_{1}(0)\right\rangle =\pi^{-1/4}e^{-\frac{1}{2%
}\left( y+L/2\right) ^{2}}\text{, or}  \notag \\
&\langle q(y)\left\vert \varphi_{1}(0)\right\rangle =\pi^{-1/4}e^{-\frac{1}{2%
}\left( y-L/2\right) ^{2}}
\end{align}
can possibly detect the phase, since local phases are meaningless in quantum
mechanics as only relative phases have any physical meaning and are
measurable. Let us again initially allow Alice to choose between one of the
following states (with $\alpha=0$ or $\alpha=\pi$) at $t=0$: 
\begin{equation}
\left\vert \varphi_{\pm}(0)\right\rangle =\frac{1}{\sqrt{2}}\left(
|0,-L/2\rangle\pm|0,+L/2\rangle\right)
\end{equation}
The global wave-function in the position basis is then%
\begin{equation}
\langle q(x)\left\vert \varphi_{\pm}(0)\right\rangle =\frac{\pi^{-1/4}}{%
\sqrt{2}}\left( e^{-\frac{1}{2}\left( x+L/2\right) ^{2}}\pm e^{-\frac {1}{2}%
\left( x-L/2\right) ^{2}}\right)
\end{equation}

We have plotted in figure \ref{initialstatecat} the probability
distributions $P_{\pm}(0)=\left\vert \langle q(x)\left\vert
\varphi_{\pm}(0)\right\rangle \right\vert ^{2}$ for both global states $%
\left\vert \varphi_{\pm}(0)\right\rangle $ at $t=0$ and for $L=10$. 
\begin{figure}[H]
\begin{center}
\includegraphics[scale=.4]{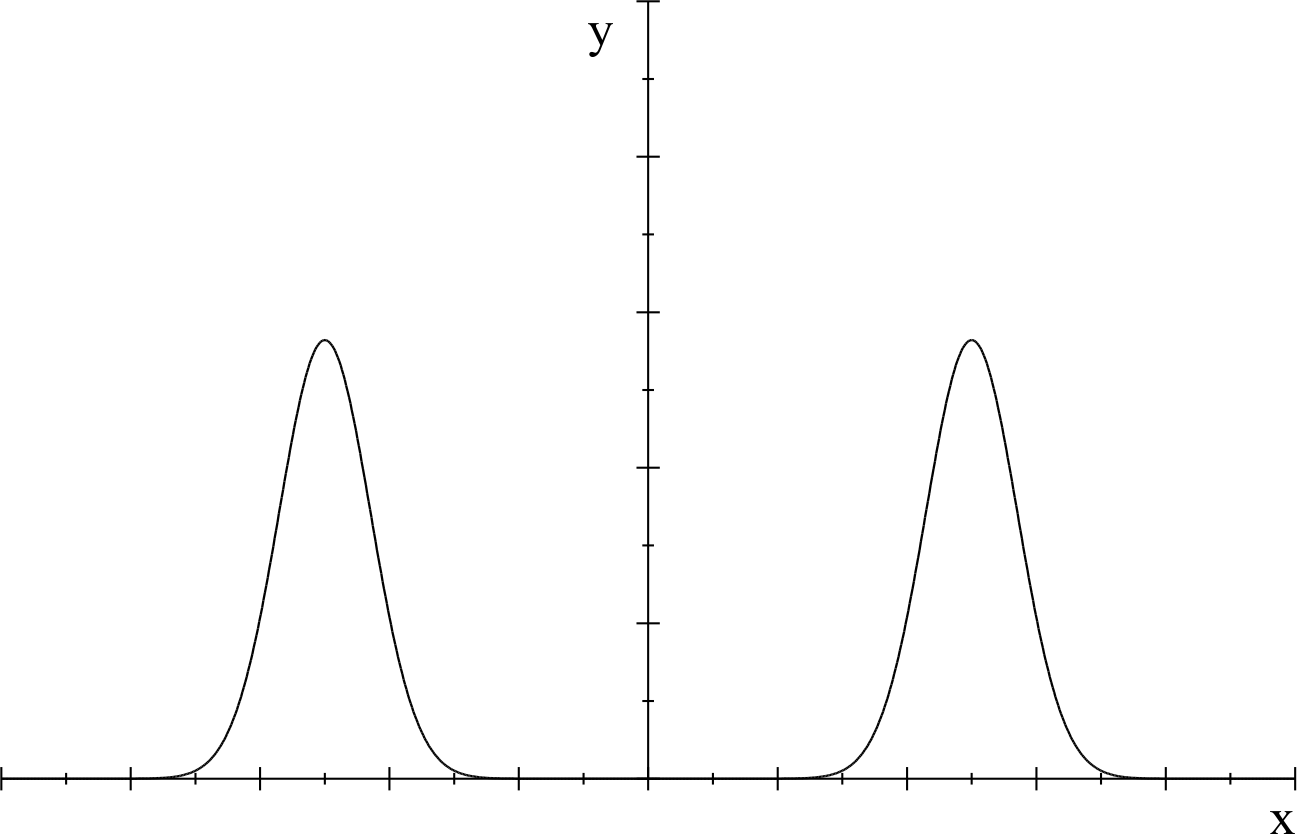}
\end{center}
\caption{{Initial Schr\"{o}dinger-Cat state}}
\label{initialstatecat}
\end{figure}

To measure the relative phase, one would need to measure the whole global
wave function non-locally. Each lump of the probability distribution would
have to be measured locally (at $t=0$) and the information would have to be
later compared. Yet, there is an easier method: one can ``drive" both
packets towards each other with a suitable Hamiltonian. Indeed, if we choose
the Hamiltonian as \eqref{generator of rotation in phase space} then the
time evolution of a coherent state is trivial: Up to a phase, a coherent
state \textit{remains} a coherent state, imitating perfectly the motion in
phase space of a classical harmonic oscillator. The best way to see this is
through (\ref{rotation in Q.phase space with fractional fourier transform}):
The time evolution of each ket of the Shr\"{o}dinger-cat superposition can
now be easily computed as a dispersion-less \textit{rotation} in phase space
with unity frequency (up to an overall physically inessential phase): 
\begin{align}
\left\vert \varphi_{1}(t)\right\rangle & =e^{-it\left( \hat{N}+\frac{1}{2}%
\hat{I}\right) }|0,-\frac{L}{2}\rangle=e^{-i\frac{t}{2}}|\frac{L}{2}\sin t,-%
\frac{L}{2}\cos t\rangle  \notag  \label{motion of a coherent state} \\
\left\vert \varphi_{2}(t)\right\rangle & =e^{-it\left( \hat{N}+\frac{1}{2}%
\hat{I}\right) }|0,\frac{L}{2}\rangle=e^{-i\frac{t}{2}}|-\frac{L}{2}\sin t,%
\frac{L}{2}\cos t\rangle
\end{align}
after a time $t=\pi/2$, \textit{both} possible global state vectors are
given by%
\begin{equation}
\left\vert \varphi_{\pm}(\pi/2)\right\rangle =\frac{1}{\sqrt{2}}e^{-i\frac {%
\pi}{4}}\left( |+\frac{L}{2},0\rangle\pm|-\frac{L}{2},0\rangle\right)
\end{equation}
where again $e^{-i\frac{\pi}{4}}$ is an unimportant overall phase factor.
The motion in phase space is depicted in figure \ref{phasespacecat} 
\begin{figure}[H]
\begin{center}
\includegraphics[scale=.2]{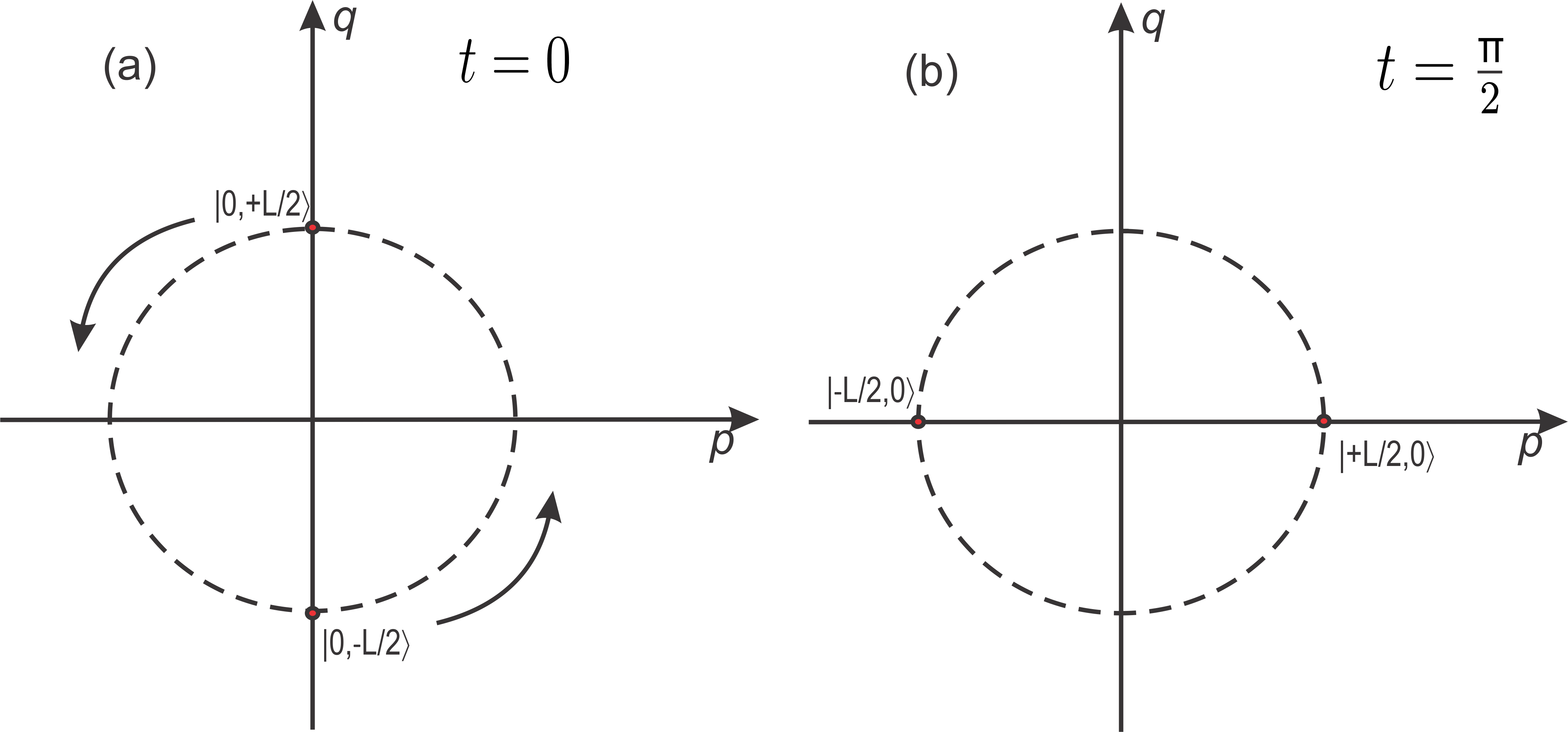}
\end{center}
\caption{{State evolution in Phase Space}}
\label{phasespacecat}
\end{figure}
In the position basis these state vectors can be written as 
\begin{equation}
\langle q(x)\left\vert \varphi_{+}(\pi/2)\right\rangle =\sqrt{2}\pi
^{-1/4}e^{-i\frac{\pi}{4}}\cos\left( \frac{Lx}{2}\right) e^{-\frac{1}{2}%
x^{2}}
\end{equation}
and 
\begin{equation}
\langle q(x)\left\vert \varphi_{-}(\pi/2)\right\rangle =\sqrt{2}i\pi
^{-1/4}e^{-i\frac{\pi}{4}}\sin\left( \frac{Lx}{2}\right) e^{-\frac{1}{2}%
x^{2}}
\end{equation}
We plot below the probability distribution $P_{+}(x)=\left\vert \langle
q(x)\left\vert \varphi_{+}(\pi/2)\right\rangle \right\vert ^{2}=\frac{2}{%
\sqrt{\pi}}\cos^{2}\left( 5x\right) e^{-x^{2}}$ for $t=\pi/2$: 
\begin{figure}[H]
\begin{center}
\includegraphics[scale=.5]{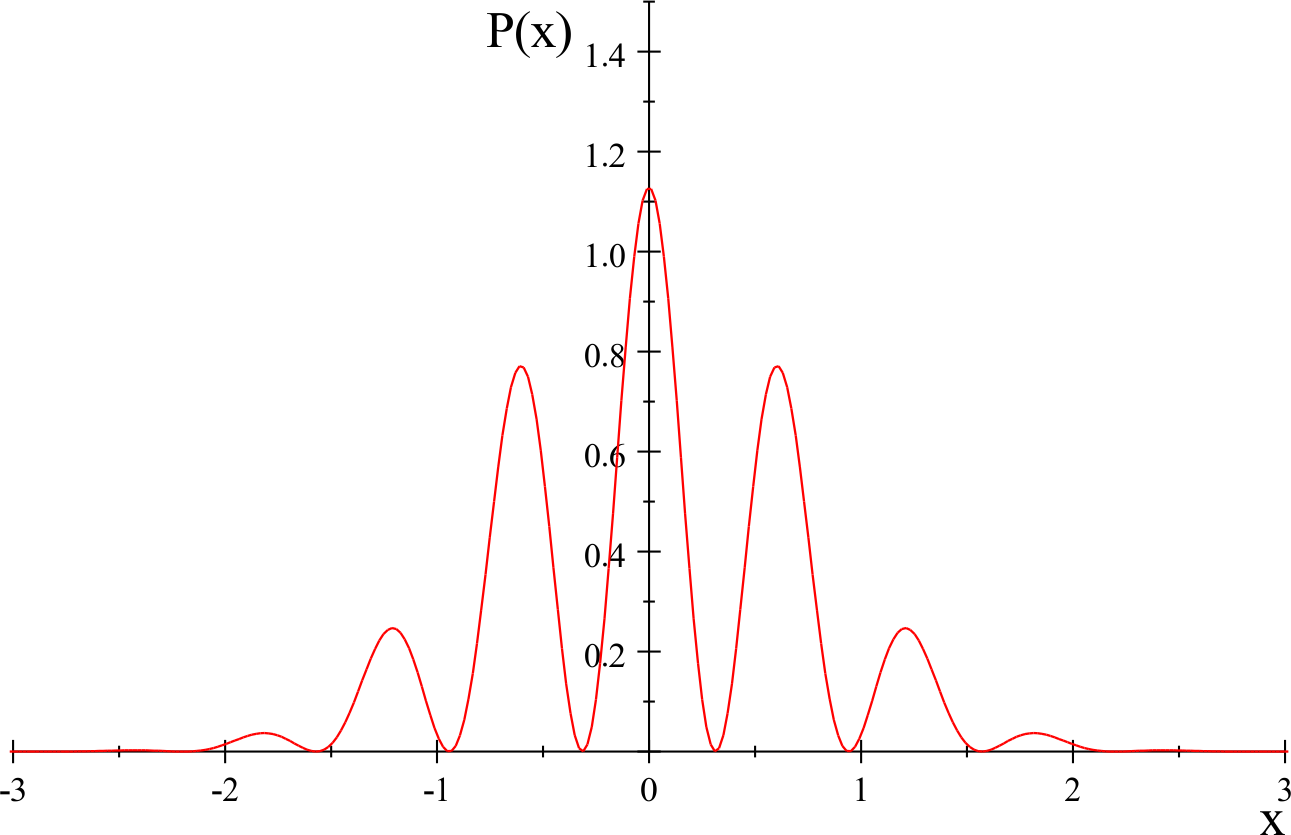}
\end{center}
\caption{$P_{+}(x)$ for $t=\protect\pi/2$ and $L=10$}
\end{figure}
We also plot below the probability distribution \newline
$P_{-}(x)=\left\vert \langle q(x)\left\vert \varphi_{-}(\pi/2)\right\rangle
\right\vert ^{2}=\frac{2}{\sqrt{\pi}}\sin^{2}\left( 5x\right) e^{-x^{2}}$
for $t=\pi/2$: 
\begin{figure}[H]
\begin{center}
\includegraphics[scale=.5]{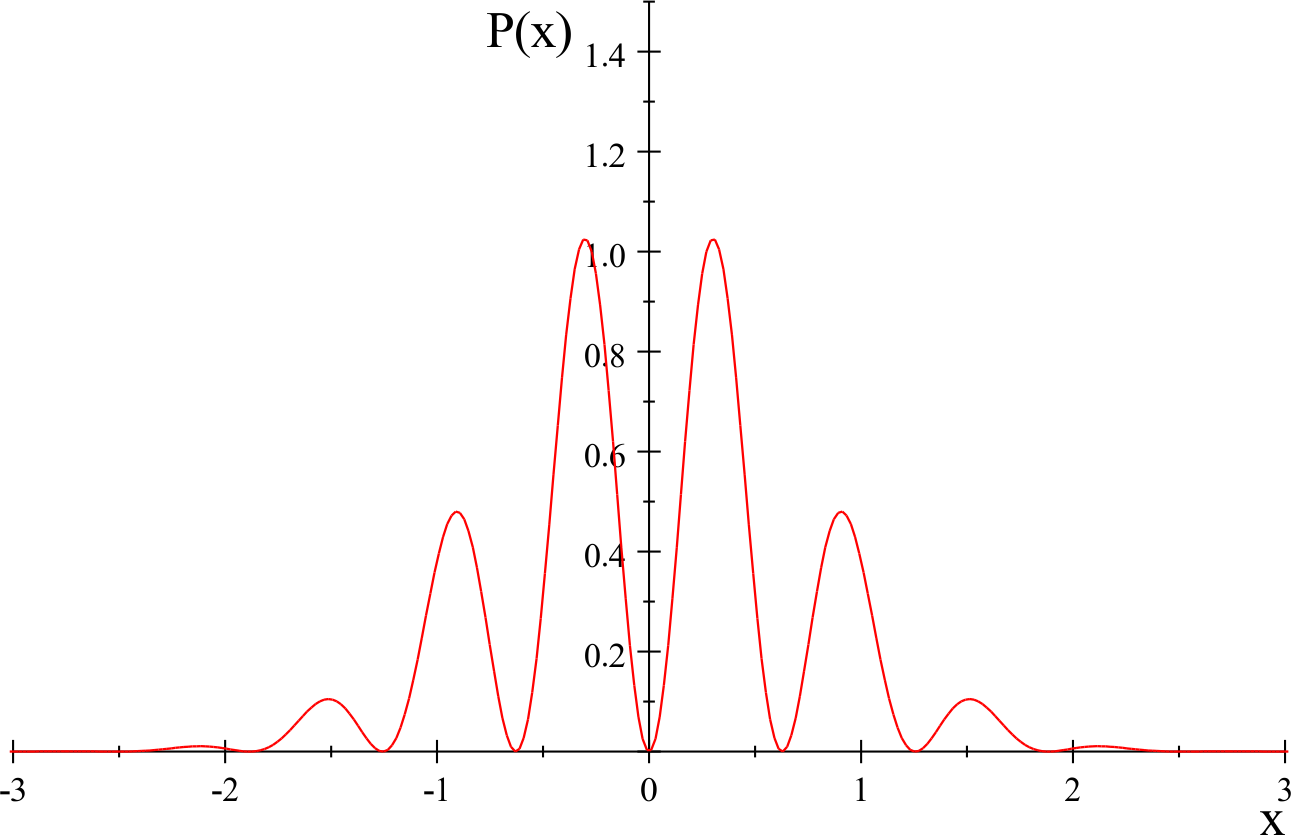}
\end{center}
\caption{$P_{-}(x)$ for $t=\protect\pi/2$ and $L=10$}
\end{figure}

Now that there is a superposition between the two pieces of wave-function,
it is easy to see the interference pattern in both cases and so one could
hope to distinguish them through local measurements. One possibility is to
use the fact that the patterns are somewhat ``complementary" in relation to
each other in the sense that the maximum for the probability density for one
pattern is the minimum for the other and vice-versa. For instance, suppose
Bob sets up particle detectors at positions $x_{n}=n\pi/5$ (for $%
n=0,\pm1,\pm2,\pm3$) at instant $t=\pi/2$ and where the \textit{precision}
of the detectors is given by $x_{n}\pm\Delta$ with $\Delta\approx0.2$. One
can estimate numerically that the probability of one of the detectors to
``click" is about $93\%$ for state $\left\vert \varphi_{+}\right\rangle $
and $35\%$ for state $\left\vert \varphi_{-}\right\rangle $. This means that
this method is \textit{unreliable} for a single measurement to distinguish
one state from the other. Bob would need to carry out a number of
repetitions of this procedure for a set of equally prepared states by Alice
in order to distinguish the states with confidence. For this reason and also
because Bob has to perform a strong position measurement which will collapse
the wave-function, we choose to coin this method as a \textit{partially}
deterministic measurement.

Suppose now that Alice prepares the cat state with an arbitrary phase $%
\alpha $\ as in (\ref{arbitrary cat state}). After a time $t=\pi /2$, the
state vector becomes 
\begin{align}
& \left\vert \varphi _{\pm }(\pi /2)\right\rangle =\frac{1}{\sqrt{2}}e^{-i%
\frac{\pi }{4}}\left( |+\frac{L}{2},0\rangle \pm e^{i\alpha }|-\frac{L}{2}%
,0\rangle \right)   \notag \\
& =\frac{1}{\sqrt{2}}e^{-i\left( \frac{\pi -2\alpha }{4}\right) }\left(
e^{-i\alpha /2}|+\frac{L}{2},0\rangle \pm e^{i\alpha /2}|-\frac{L}{2}%
,0\rangle \right) 
\end{align}%
and the position basis wave-function can be written as%
\begin{align}
& \langle q(x)\left\vert \varphi _{+}(\pi /2)\right\rangle =  \notag \\
& =\sqrt{2}\pi ^{-1/4}e^{-i\left( \frac{\pi -2\alpha }{4}\right) }\cos \left[
\frac{1}{2}\left( Lx-\alpha \right) \right] e^{-x^{2}/2}
\end{align}%
\begin{align}
& \langle q(x)\left\vert \varphi _{-}(\pi /2)\right\rangle =  \notag \\
& =i\sqrt{2}\pi ^{-1/4}e^{-i\left( \frac{\pi -2\alpha }{4}\right) }\sin %
\left[ \frac{1}{2}\left( Lx-\alpha \right) \right] e^{-x^{2}/2}
\end{align}%
If Alice provides Bob with a sufficiently large ensemble of equally prepared
states, he can find the value of $\cos \alpha $ (or $\sin \alpha $) with
arbitrary precision by measuring the particles position with a detector
localized in $x=0$. Again we have a \textit{partially} deterministic
measurement of phase $\alpha $. It is important to note that, though one has
in principle an infinite dimensional Hilbert space corresponding to the 1D
motion of a quantum particle, since one is actually measuring a phase
difference between two orthogonal states, in a sense one is effectively
restrained to a two-level system (a qubit). Where can one actually find this
qubit? The answer lies within the modular variable concept that we review in
the next section.

\section{Aharonov's Modular Variables and Schwinger's Formalism for Finite
Quantum Mechanics}

Though the space translation operator $\hat{V}_{L}=e^{i\hat{P}L}$ is not
hermitean and therefore it is \textit{not} a genuine quantum mechanical
observable, one can consider the \textquotedblleft phase" of such operator
which is exactly the modular momentum. It is \textit{modular} because it is
clearly defined up to a value $p+\frac{2n\pi }{L}$ with $n\in Z$.
Analogously, we can define a \textit{modular position} variable as the phase
of the momentum space translator $\hat{U}_{\frac{2\pi }{L}}=e^{\frac{2\pi i}{%
L}\hat{Q}}$ (with modular position given up to a value $q+nL$). It follows
immediately from (\ref{continuous weyl commutation relation}) that $\hat{V}%
_{L}$ and $\hat{U}_{\frac{2\pi }{L}}$ \textit{commute}, so there are states
that are simultaneously eigenkets \textit{both} of modular momentum and
modular position. For an apparatus formed by a lattice of a very large
number $N$ of slits (with period $L$), one expects that $\phi (x)$ should be
an $L$-periodic function and in this case, (\ref{dynamical non-local
equation}) implies that the modular momentum is \textit{exactly} conserved.
Consider the paradigmatic experiment of diffraction of a particle through
such an apparatus with a large set of $N$ equidistant slits as in Figure \ref%
{fignslit}. 
\begin{figure}[H]
\centering
\includegraphics[
natheight=1.760800in,
natwidth=3.146200in,
height=1.7979in,
width=3.1903in
]%
{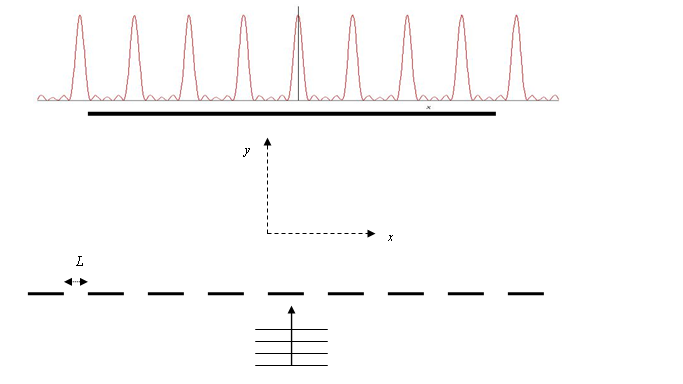}%
\caption{n-slit interference experiment.}\label{fignslit}
\end{figure}

One can model the interaction of the particle with the slits through the
following Hamiltonian in the $x$ direction:%
\begin{equation}
\hat{H}(t)=\frac{1}{2m}\hat{P}^{2}+V(\hat{Q})\delta (t)\text{, with}\quad V(%
\hat{Q}+L)=V(\hat{Q})
\end{equation}%
where the particle \textquotedblleft hits the screen\textquotedblright\ at $%
t=0$. The fundamental physical assumptions here are that the interaction of
the particle happens so fast that the unitary time evolution is given by $%
\hat{U}(t)\approx e^{-iV(\hat{Q})}$. By expanding this function in a Fourier
series one gets%
\begin{equation}
e^{-iV(\hat{Q})}=\sum\limits_{n\in Z}c_{n}e^{\frac{2\pi in}{L}\hat{Q}%
}=\sum\limits_{n\in Z}c_{n}\hat{U}_{\frac{2\pi n}{L}}
\end{equation}%
The initial state-vector (in the $x$ direction) of the particle is an
eigenstate of momentum with zero momentum $\left\vert \psi
_{(I)}\right\rangle =\left\vert p(0)\right\rangle $ so the state just after
the interaction becomes $\left\vert \psi _{(F)}\right\rangle =\hat{U}%
(t)\left\vert p(0)\right\rangle $. The final state is then given by%
\begin{equation}
|\psi _{(F)}\rangle =\sum\limits_{n\in Z}c_{n}\hat{U}_{\frac{2\pi n}{L}%
}\left\vert p(0)\right\rangle =\sum\limits_{n\in Z}c_{n}|p(\frac{2\pi n}{L}%
)\rangle 
\end{equation}%
The resulting state has indeed the remarkable property of being an
eigenstate \textit{both} of $\hat{U}_{2\pi /L}$ and $\hat{V}_{L}$. That $%
\left\vert \psi _{(F)}\right\rangle $ is indeed an eigenstate of $\hat{V}_{L}
$ follows directly from \eqref{eigenvalue eq. for translators}. That the
same state \textit{must also} necessarily be an eigenstate of $\hat{U}_{2\pi
/L}$ follows\ from the Weyl relation \eqref{continuous weyl commutation
relation}. Since $\hat{U}_{2\pi /L}$ and $\hat{V}_{L}$ commute and are also
unitary, their eigenvalues are necessarily complex phases. This why Aharonov
and collaborators coined these phases as \textit{modular variables}. A phase
space description of such a state is given by Figure \ref{figexamplemod}. 
\begin{figure}[H]
\centering
\includegraphics[
natheight=1.760800in,
natwidth=2.572800in,
height=1.7979in,
width=2.6143in
]%
{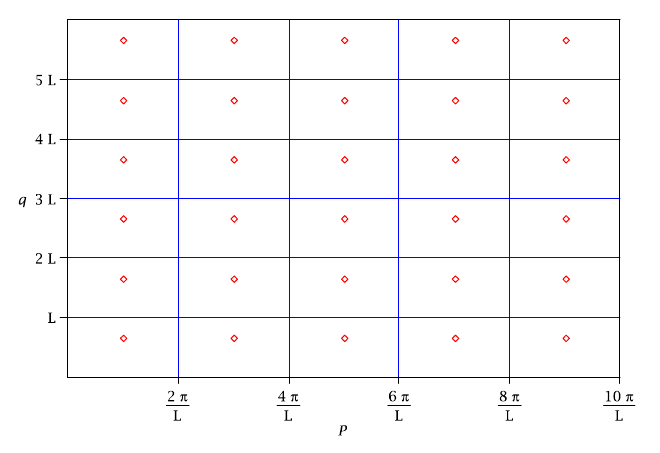}%
\caption{{State with $q\ mod=2L/3$ and $p mod=\pi/L$.}}\label{figexamplemod}
\end{figure}
This means that for the state represented above, one has that, in each cell,
it is represented by an exact point with sharp values of modular position
and momentum but there is a complete uncertainty about \textit{which} cell
it belongs to. This is a fundamental feature of the modular variable
description. States with this peculiar mathematical structure have been also
described by Zak to study systems with periodic symmetry in quantum
mechanics \cite{zak1967,zak1968}. We may call them Aharonov-Zak states (AZ).
Notice that the equation of motion (\ref{dynamical non-local equation}) of
the modular momentum is clearly non-local and one can perceive that this
kind of \textquotedblleft dynamical non-locality" can be naturally expressed
within the Heisenberg picture. It would be almost impossible to capture this
idea within the Schr\"{o}dinger picture. Still one may address the question
of how to describe explicitly the eigenkets of modular variables in a Schr%
\"{o}dinger-like picture where one pays primary attention to the motion of
the kets instead of the observables. The quantum vector space of the states
of modular variables must be a finite-dimensional subspace (or at least a
discrete subspace) of the infinite dimensional (continuous) space of states
of the quantum motion of the particle in the $x$ direction. One possible
approach to this description was given in \cite{aharonovpendletonpetersen}
which we briefly review below.

\subsection{Aharonov%
\'{}%
s Modular qubit}

Consider the two-slit apparatus mentioned above. In this case, one should
expect that the modular variables form a two-dimensional space (a spin $1/2$
algebra or a qubit algebra). One may define spin-like operators:%
\begin{align}
\hat{\sigma}_{3}& =\frac{1}{2i}(\hat{U}_{\frac{\pi }{L}}-\hat{U}_{\frac{\pi 
}{L}}^{\dagger })  \label{Aharonov modular spin operators} \\
\hat{\sigma}_{1}& =\frac{1}{2}(\hat{V}_{L}+\hat{V}_{L}^{\dagger })-\frac{1}{2%
}(\hat{V}_{L}-\hat{V}_{L}^{\dagger })\hat{\sigma}_{3}  \notag \\
\hat{\sigma}_{2}& =-\frac{i}{2}(\hat{V}_{L}-\hat{V}_{L}^{\dagger })+\frac{i}{%
2}(\hat{V}_{L}+\hat{V}_{L}^{\dagger })\hat{\sigma}_{3}  \notag
\end{align}%
acting upon a \textit{two-dimensional} subspace of vectors given by $%
W^{(2)}=\left\{ |q(+L/2)\rangle ,|q(-L/2)\rangle \right\} $ from the \textit{%
infinite-dimensional} (continuous) space $W^{(\infty )}=\left\{ |q(x)\rangle
,x\in 
\mathbb{R}
\right\} $. By using (\ref{continuous weyl commutation relation}) one can
indeed check that the above operators obey the usual algebra for the Pauli
matrices. 
\begin{equation}
\lbrack \hat{\sigma}_{j},\hat{\sigma}_{k}]=2i\hat{\sigma}_{l}\varepsilon
_{jk}^{l}\qquad \text{and\qquad }\left\{ \hat{\sigma}_{j},\hat{\sigma}%
_{k}\right\} =2\hat{I}\delta _{jk}  \label{spin algebra}
\end{equation}%
when the operators defined in (\ref{Aharonov modular spin operators}) are
restricted to act upon $W^{(2)}$%
\begin{equation}
\hat{\sigma}_{3}|q(\pm L/2)\rangle =\pm |q(\pm L/2)\rangle
\end{equation}%
so that the two sharp position eigenstates at \textit{each slit} generate
the modular qubit and they also diagonalize the $\hat{\sigma}_{3}$ operator.
In this way, one can say that the \textit{full} infinite-dimensional space $%
W^{(\infty )}$ is the \textit{direct sum} $W^{(\infty )}=W^{(\infty \prime
)}\oplus W^{(2)}$ of the Aharonov qubit $W^{(2)}$ with the (also
infinite-dimensional space) $W^{(\infty \prime )}$=$\{|q(x)\rangle $, $x\in 
\mathbb{R}
-\left\{ \pm L/2\}\right\} $. We propose below a \textit{different} approach
to construct explicitly a qudit for modular variables based on Schwinger's
finite quantum kinematics.

\subsection{The Schwinger-based Approach to the Modular Variable qubit}

(The main references for this subsection are \cite{schwinger}, \cite%
{schwinger1960})

\subsubsection{Schwinger's Quantum Kinematics:}

Let $W%
{\acute{}}%
^{(N)}$ be an $N$-dimensional quantum space generated by an orthonormal
basis $\{\left\vert u_{k}\right\rangle \}$, $(k=0,1,\ldots N-1)$ which means
that 
\begin{equation}
\vert u_{k}\rangle\langle u^{k}\vert=\hat{I}\qquad\text{and}\qquad\langle
u^{j}\vert u_{k}\rangle=\delta_{k}^{j}
\end{equation}
These are considered to be finite position states. We also define a unitary
translation operator $\hat{V}$ that acts cyclically over this basis:%
\begin{equation}
\hat{V}\left\vert u_{k}\right\rangle =\left\vert u_{k-1}\right\rangle
\end{equation}
The cyclicity implies that the following periodic boundary condition must be
obeyed:%
\begin{equation*}
\left\vert u_{k+N}\right\rangle =\left\vert u_{k}\right\rangle
\end{equation*}
Which means that%
\begin{equation}
\hat{V}^{N}=\hat{I}
\end{equation}
so that the eigenvalues of $\hat{V}$ consists of the $N^{th}$ roots of unity:%
\begin{equation}
\hat{V}\left\vert v_{k}\right\rangle =v_{k}\left\vert v_{k}\right\rangle
\qquad\text{with}\qquad(v_{k})^{N}=1
\end{equation}
The set $\{\left\vert v_{k}\right\rangle \},(k=0,1,2,\ldots N-1)$ is also an
orthonormal basis of $W%
{\acute{}}%
^{(N)}$ (the finite set of momentum states) and the N distinct eigenvalues
are explicitly given by%
\begin{equation}
v_{k}=e^{\frac{2\pi i}{N}k}
\end{equation}
With a convenient choice of phase, one can show that 
\begin{equation*}
\langle u^{j}\vert v_{k}\rangle=\frac{1}{\sqrt{N}}e^{\frac{2\pi i}{N}jk}=%
\frac{1}{\sqrt{N}}v_{k}^{j},
\end{equation*}
which is a finite analog of the plane wave function. It is not difficult to
convince oneself of the following property:%
\begin{equation}
\frac{1}{N}\sum \limits_{l=0}^{N-1}e^{\frac{2\pi i}{N}(j-k)l}=\delta_{jk}
\end{equation}

One may then define a unitary translator $\hat{U}$ that acts cyclically upon
the momentum basis analogously as was carried out for the $\hat{V}$ operator:%
\begin{equation*}
\hat{U}\left\vert v_{k}\right\rangle =\left\vert v_{k+1}\right\rangle
\end{equation*}%
The same analysis applies now to the $\hat{U}$ operator and amazingly, the
eigenstates of $\hat{U}$ can be shown to coincide with the original finite
set of position basis with the same spectrum of $\hat{V}$:%
\begin{equation*}
\hat{U}\left\vert u_{k}\right\rangle =v_{k}\left\vert u_{k}\right\rangle
\end{equation*}%
The finite index set $j,k=0,1,2,\ldots ,N-1$ takes values in the finite ring 
$Z_{N}$ of $mod\ N$ integers. When $N=p$ is a prime number, $Z_{p}$ has the
structure of a \textit{finite field}. One distinguished property that is not
difficult to derive is the well-known Weyl commutation relations between
powers of the unitary translator operators which is a finite analog of 
\eqref{continuous weyl
commutation relation}. 
\begin{equation}
\hat{V}^{j}\hat{U}^{k}=v^{jk}\hat{U}^{k}\hat{V}^{j}
\end{equation}

In the next subsection, we carry out the continuous limit of Schwinger's
finite structure in order to present our proposal of a finite analogue of
Aharonov's modular variable concept and we also discuss the concept of 
\textit{pseudo-degrees of freedom} \cite{lobonemes1995}, an idea that we
believe to be essential in order to capture the true nature of the modular
variables discrete Hilbert space.

\subsection{Heuristic Continuum Limit}

The implementation of the ``continuum heuristic limit'' (when the
dimensionality of the quantum spaces approach infinity) can be performed in
two distinct manners: one symmetric and the other non-symmetric between the
position and momentum states. First, we briefly outline below, the symmetric
case and following this, we present the non-symmetric limit which we shall
use to discuss the modular variable concept.

\subsubsection{The Symmetric Continuum Limit}

Let the dimension $N$ of the quantum space be an odd number (with no loss of
generality) and let us re-scale in equal footing the finite position and
momentum states as%
\begin{equation}
\left\vert q(x_{j})\right\rangle =\left( \frac{N}{2\pi}\right)
^{1/4}\left\vert u_{j}\right\rangle\text{ and }\left\vert
p(y_{k})\right\rangle =\left( \frac{N}{2\pi}\right) ^{1/4}\left\vert
v_{k}\right\rangle
\end{equation}
with%
\begin{equation}
x_{j}=\left( \frac{2\pi}{N}\right) ^{1/2}j\qquad\text{and}\qquad
y_{k}=\left( \frac{2\pi}{N}\right) ^{1/2}k
\end{equation}
so that the discrete indices are disposed symmetrically in relation to zero:%
\begin{align}
&j,k=-\frac{(N-1)}{2},\ldots,+\frac{(N-1)}{2} \quad\text{ and}  \notag \\
&\Delta x_{j}=\Delta y_{k}=\left( \frac{2\pi}{N}\right) ^{1/2}
\end{align}
We may write the completeness relations for both basis as%
\begin{align}
\hat{I}&=\sum \limits_{j=-\frac{(N-1)}{2}}^{+\frac{(N-1)}{2}}\left( \frac{%
2\pi}{N}\right) ^{1/2}\left\vert q(x_{j})\right\rangle \left\langle
q(x_{j})\right\vert  \notag \\
&= \sum \limits_{k=-\frac{(N-1)}{2}}^{+\frac{(N-1)}{2}}\left( \frac{2\pi}{N}%
\right) ^{1/2}\left\vert p(y_{k})\right\rangle \left\langle
p(y_{k})\right\vert
\end{align}
One can give a natural heuristic interpretation of the $N\rightarrow\infty$
limit for the above equation as%
\begin{equation}
\hat{I}=\int \limits_{-\infty}^{+\infty}dx\left\vert q(x)\right\rangle
\left\langle q(x)\right\vert =\int \limits_{-\infty}^{+\infty}dy\left\vert
p(y)\right\rangle \left\langle p(y))\right\vert
\end{equation}
The generalized orthonormalization relations show that the new defined basis
is formed by singular state-vectors with \textquotedblleft infinite
norm\textquotedblright:%
\begin{align}
\left\langle q(x_{j})\right\vert q(x_{k})\rangle&=\left\langle
p(x_{j})\right\vert p(x_{k})\rangle  \notag \\
&=\underset{N\rightarrow\infty}{lim}\left( \frac{N}{2\pi}\right)
^{1/2}\delta_{jk}\rightarrow\delta(x_{j}-x_{k})
\end{align}
The above equation also serves as a heuristic-based definition of the Dirac
delta function as a continuous limit of the discrete Kronecker delta. Yet
the overlap of elements from the distinct basis gives us the well-known
plane-wave basis as expected:%
\begin{equation}
\left\langle q(x_{j})\right\vert p(y_{k})\rangle=\frac{1}{\sqrt{2\pi}}%
e^{ix_{j}y_{k}}\qquad(\hbar=1)  \label{plane wave function 3}
\end{equation}

\subsubsection{The Non-Symmetric Continuum Limit}

We introduce now a different scaling for the variables of position and
momentum with a given $\xi\in%
\mathbb{R}
$:%
\begin{equation}
\left\vert q(x_{j})\right\rangle =\left( \frac{N}{\xi}\right) ^{1/
2}\left\vert u_{j}\right\rangle \text{ and }\left\vert p(y_{k})\right\rangle
=\left( \frac{\xi}{2\pi}\right) ^{1/ 2}\left\vert v_{k}\right\rangle
\end{equation}
with%
\begin{equation}
x_{j}=\frac{\xi}{N}j\qquad\text{and}\qquad y_{k}=\frac{2\pi}{\xi}k
\end{equation}
so that 
\begin{equation}
\Delta x_{j}=\frac{\xi}{N}\rightarrow0\qquad\text{for}\qquad N\rightarrow
\infty
\end{equation}
Only the position states become singular and the $x_{j}$ variable takes
value in a bounded quasi-continuum set so that the resolution of identity
can be written as%
\begin{equation}
\hat{I}_{\xi}=\int \limits_{-\xi/ 2}^{+\xi/ 2}dx\left\vert q(x)\right\rangle
\left\langle q(x)\right\vert =\frac{2\pi}{\xi}\sum
\limits_{k=-\infty}^{+\infty}\left\vert p(y_{k})\right\rangle \left\langle
p(y_{k})\right\vert
\end{equation}
Note that the momentum states continue to be of finite norm and their sum is
taken over the enumerable but discrete set of integers $k\in Z$. The
identity operator $\hat{I}_{\xi}$ can be thought as the projection operator
on the subspace of periodic functions with period $\xi$. The overlap between
position and momentum states is again given by the usual plane-wave function %
\eqref{plane wave function 3}. Note also that we could have reversed the
above procedure by choosing from the beginning the opposite scaling for the
position and momentum states. In this case, the momentum eigenkets would
form a continuous bounded set of singular state-vectors and the position
eigenvectors would form an enumerable infinity of finite norm kets.

\subsubsection{Pseudo Degrees of Freedom and the finite analog of the
Aharonov-Zak states}

In the $x^{1}-x^{2}$ plane, one can easily visualize the translations of the
ket $\left\vert q(\vec{x})\right\rangle $ acted repeatedly upon with $\hat {V%
}_{\vec{\xi}}$ as in Figure \ref{figtwodegreesfreedom}, where the resulting
position kets can be represented on a straight line in the plane that
contains point $\vec{x}$ but with slope given by the $\vec{\xi}$ direction.%
\begin{figure}[H]
\centering
\includegraphics[
natheight=1.760800in,
natwidth=2.405900in,
height=1.7979in,
width=2.4474in
]%
{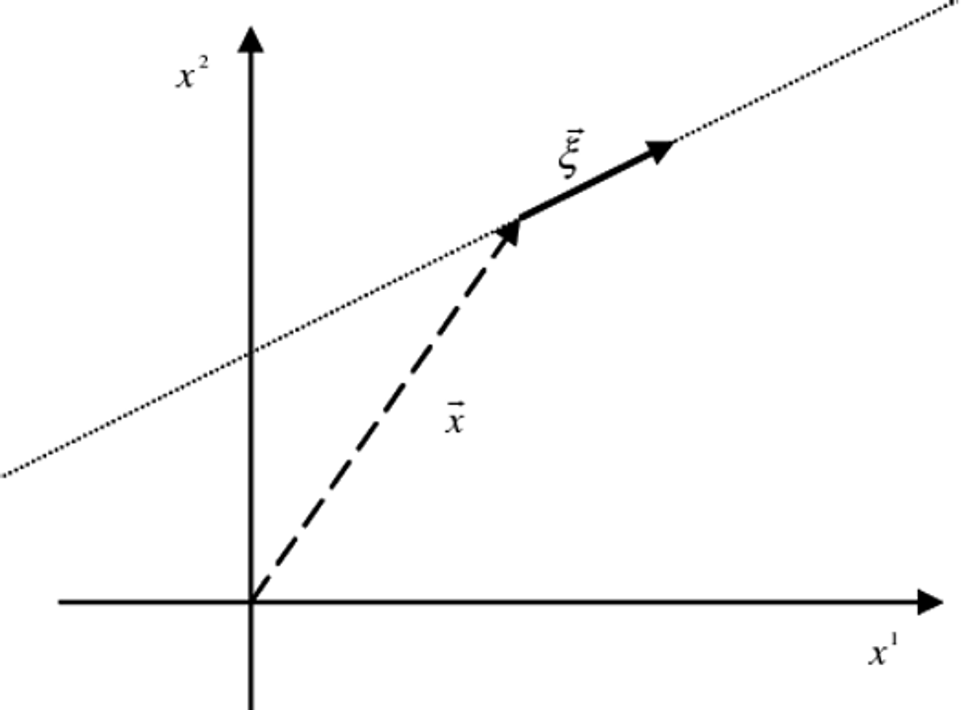}%
\caption{Two degrees of freedom.}\label{figtwodegreesfreedom}
\end{figure}
Of course, to reach an arbitrary point in the plane, one needs at least two
linear independent directions. This is precisely what one means when it is
said that the plane is two-dimensional. But things for finite quantum spaces
are not quite so simple. Let us consider first a 4-dimensional system given
by the product of two 2-dimensional spaces (two qubits) $W^{(4)}=W^{(2)}%
\otimes W^{(2)}$ (it is important to notice here that one must not confuse
the dimension of space, the so called degree of freedom with the
dimensionality of the quantum vector spaces). We shall discard in the
following discussion, the indices that indicate dimensionality to eliminate
excessive notation. So let $\{\left\vert u_{0}\right\rangle ,\left\vert
u_{0}\right\rangle \}$ be the position basis for each individual qubit space
so that computational (unentangled) basis of the tensor product spaces is $%
\{\left\vert u_{0}\right\rangle \otimes\left\vert u_{0}\right\rangle
,\left\vert u_{0}\right\rangle \otimes\left\vert u_{1}\right\rangle
,\left\vert u_{1}\right\rangle \otimes\left\vert u_{0}\right\rangle
,\left\vert u_{1}\right\rangle \otimes\left\vert u_{1}\right\rangle \}$. One
may represent such finite 2-space as the discrete set formed by the four
points depicted in Figure \ref{finite2qubits}: 
\begin{figure}[H]%
\centering
\includegraphics[
natheight=1.760800in,
natwidth=1.823000in,
height=1.7979in,
width=1.8602in
]%
{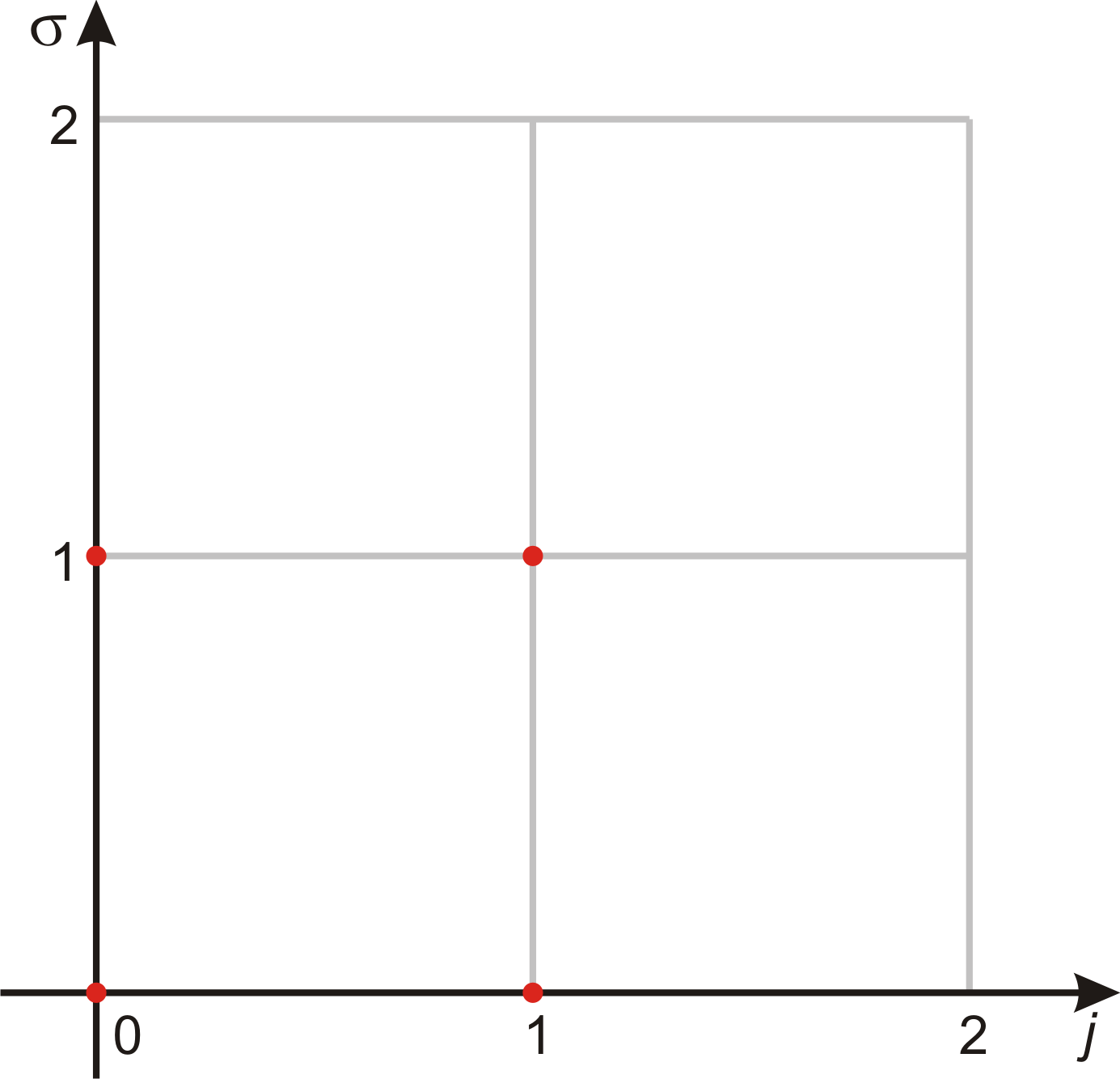}%
\caption{Finite 2-space for 2 qubits.}\label{finite2qubits}
\end{figure}

One may even construct distinct ``straight lines'' in this discrete
two-dimensional space acting upon the computational basis $\{\left\vert
u_{j}\right\rangle \otimes\left\vert u_{k}\right\rangle \}$ with the $\hat {V%
}\otimes\hat{V}$ operator as shown in Figure \ref{figlinesdiscrete}.
\begin{figure}[H]%
\centering
\includegraphics[
natheight=1.760800in,
natwidth=1.833400in,
height=1.7979in,
width=1.8715in
]%
{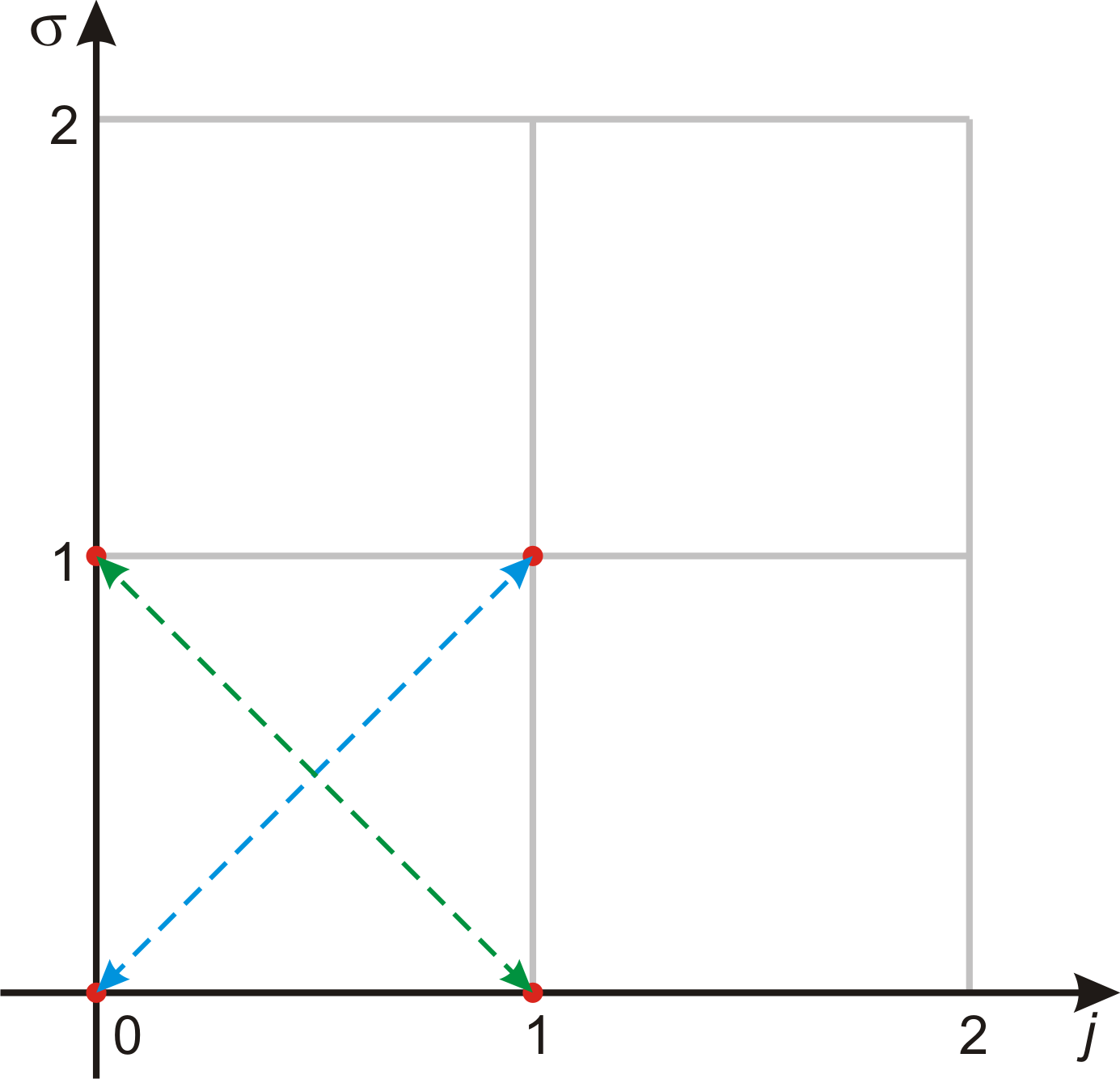}%
\caption{{Discrete parallel lines (0, 0); (1, 1) and (0,1);(1,0).}}\label{figlinesdiscrete}
\end{figure}
Each of the two parallel ``straight lines'' are geometric invariants of the discrete 2-plane under
the action of $\hat{V}\otimes\hat{V}$.

Consider now, a six-dimensional quantum space $W^{(6)}$ = $W^{(2)}\otimes
W^{(3)}$ given by the product of a qubit and a qutrit with finite position
basis given by respectively $\{\left\vert u_{0}\right\rangle ,\left\vert
u_{1}\right\rangle \}$ and $\{\left\vert u_{0}\right\rangle ,\left\vert
u_{1}\right\rangle ,\left\vert u_{2}\right\rangle \}$. In this case, the
fact the dimensions of the individual are \textit{coprime} means that the
action of the $\hat{V}\otimes\hat{V}$ operator on the product basis $%
\{\left\vert u_{j}\right\rangle \otimes\left\vert u_{\sigma}\right\rangle
\},\ j=0,1$ and $\sigma=0,1,2$ can be identified with the action of $\hat{V}%
^{(6)}=\hat {V}^{(2)}\otimes\hat{V}^{(3)}$ on the same basis relabeled as $%
\{\left\vert u_{0}\right\rangle ,\left\vert u_{1}\right\rangle ,\left\vert
u_{2}\right\rangle ,\left\vert u_{3}\right\rangle ,\left\vert
u_{4}\right\rangle ,\left\vert u_{5}\right\rangle \}$. One can start with
the $\left\vert u_{0}\right\rangle \otimes\left\vert u_{0}\right\rangle $
state and cover the whole space with one single line as shown in Figure \ref%
{figsingleline}.
\begin{figure}[H]%
\centering
\includegraphics[
natheight=1.760800in,
natwidth=2.166400in,
height=1.7979in,
width=2.2061in
]%
{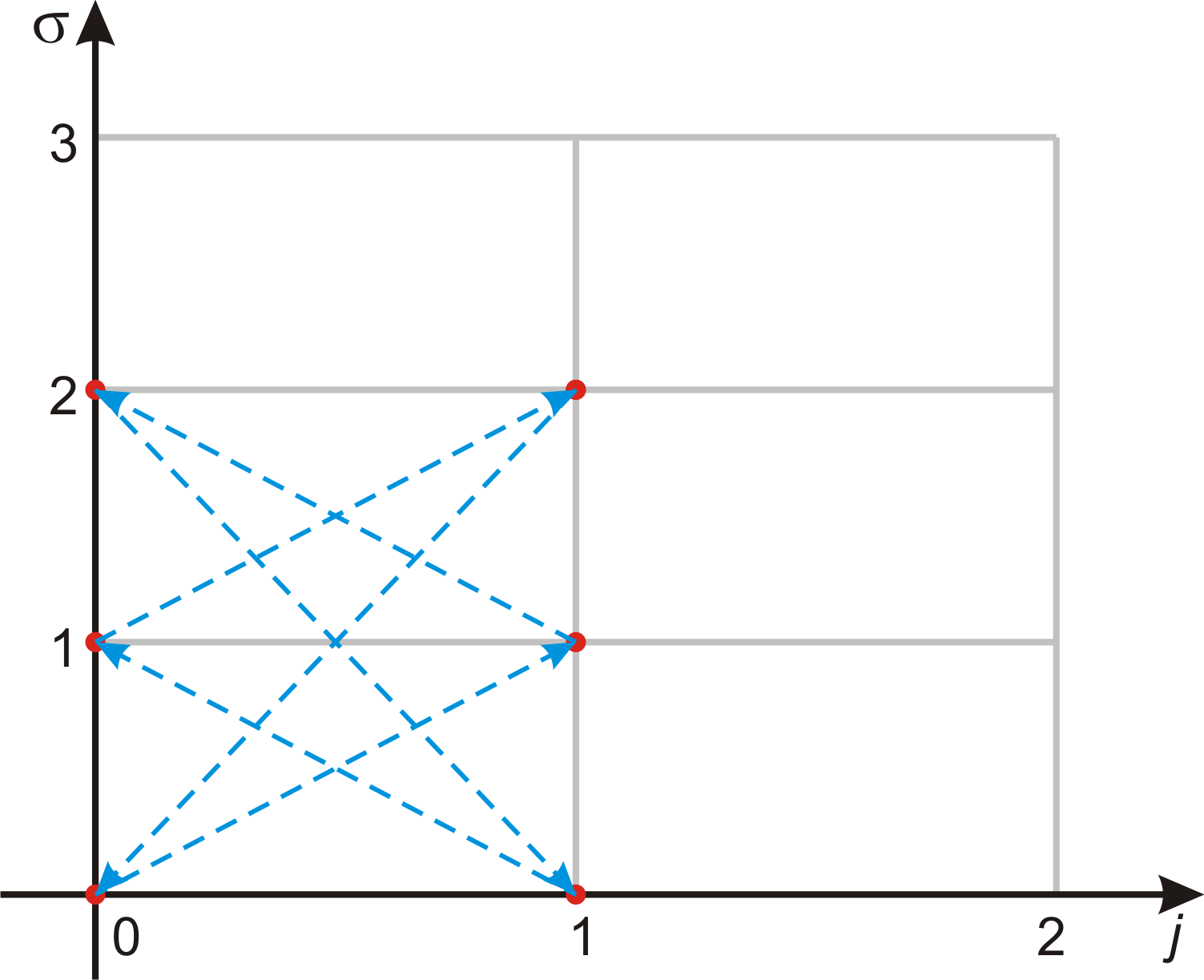}%
\caption{{A single line covers the whole space.}}\label{figsingleline}
\end{figure}

This reduction of two degrees of freedom to only one single effective degree
of freedom is a general fact for all product spaces when the dimensions of
the factor spaces are co-prime. This fact follows from elementary number
theory and it can be shown that when $gcd(N_{a},N_{b})=1$ for quantum spaces 
$W^{(N_{a})}$ and $W^{(N_{b})}$ then one may say that the two degrees of
freedom are actually pseudo-degrees of freedom because one can associate 
\textit{only one effective} single degree of freedom to the system in this
sense. In the case above, it is easy to see that one of the pseudo degrees
of freedom is nothing else but a factorizing period of the larger space. In
this way we can either interpret the above finite six-dimensional position
space as three periods of two or as two periods of three. All this is rather
obvious and elementary, but surprisingly this is precisely the kind of
mathematical structure that we propose to be behind Aharonov's concept of
modular variables. We can give a precise mathematical description of a
finite analogue of this phenomenon in terms of the pseudo-degrees of
freedom: consider $W^{(N)}=W^{(N_{a})}\otimes W^{(N_{b})}$ as the state
space for a quantum mechanical system with $gcd(N_{a},N_{b})=1$. We can then
offer an interpretation for this single degree of freedom of $W^{(N)}$ as a
degree composed of \textquotedblleft$N_{b}$ periods of size $N_{a}$%
\textquotedblright\ (or vice-versa). In fact, we may define the following
state of $W^{(N)}$:%
\begin{equation*}
\left\vert j_{a},\sigma_{b}^{(N)}\right\rangle =\left\vert
v_{j_{a}}{}^{(N_{a})}\right\rangle \otimes\left\vert
u_{\sigma_{b}}{}^{(N_{b})}\right\rangle
\end{equation*}
This state is simultaneously an eigenstate of finite momentum of $%
W^{(N_{a})} $ and finite position of $W^{(N_{b})}$and clearly represents a
finite analogue of the state represented in Figure 10. They are also
simultaneous eigenstates of the (obviously commuting operators since they
act on different spaces) unitary operators $\hat{V}^{(N_{a})}\otimes\hat{I}%
^{(N_{b})}$ and $\hat {I}^{(N_{a})}\otimes\hat{U}^{(N_{b})}$. Let us
illustrate this again with an example of our \textquotedblleft toy
six-dimensional\textquotedblright\ case: Let the state $\left\vert
1,2^{(6)}\right\rangle =\left\vert v_{1}{}^{(2)}\right\rangle
\otimes\left\vert u_{2}{}^{(3)}\right\rangle $ be represented by the phase
space plot (figure \ref{figexamplemod2}).
\begin{figure}[H]%
\centering
\includegraphics[
natheight=1.760800in,
natwidth=2.156000in,
height=1.7979in,
width=2.1958in
]%
{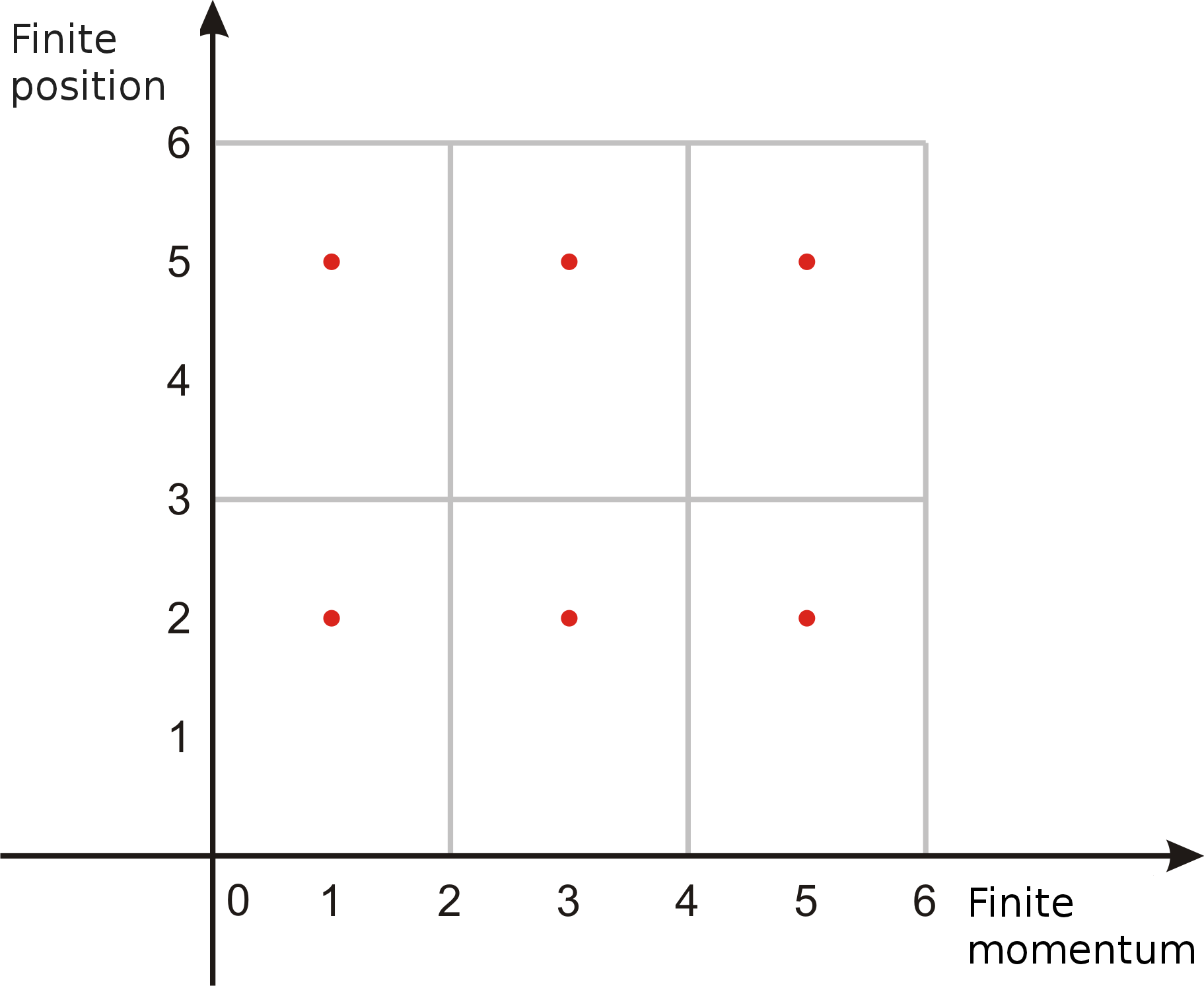}%
\caption{State $\left\vert 1,2^{(6)}\right\rangle \in W^{\left(6\right)  }$.}\label{figexamplemod2}
\end{figure}
States with this peculiar mathematical structure have been also described by
Zak to study systems with periodic symmetry in Quantum Mechanics \cite%
{zak1967,zak1968}. We may call them Aharonov-Zak states (AZ). The particular
AZ state obtained in the $n$-slit (with large enough $n$) can be thought as
obtained by an ideal projective measurement performed by the slit apparatus
on the second modular variable subspace of the particle:%
\begin{align}
&|v_{0}^{(N_{a})}\rangle\otimes|v_{0}^{(N_{b})}\rangle \overset{projective
``space" measurement}{\underset{in\ the\ \sec ond\ subspace}{\Longrightarrow}%
}  \notag \\
&|v_{0}^{(N_{a})}\rangle\otimes|u_{0}^{(N_{b})}\rangle =\frac{1}{\sqrt{N_{b}}%
}\sum
\limits_{\sigma_{b}=0}^{N_{b}-1}|v_{0}^{(N_{a})}\rangle\otimes|v_{%
\sigma_{b}}^{(N_{b})}\rangle
\end{align}

The continuum limit of the AZ state can be constructed through the
non-symmetric limit discussed in a previous subsection. The only care that
must be taken is that, given the two subspaces, the opposite limit must be
taken for each subspace. That is, if one chooses to make the momentum basis
of the first subspace go to infinity as a bounded continuum, then for the
second subspace it is the position basis that must become a bounded
continuum set and vice-versa.

\section{Concluding Remarks}

We presented a compact review of some phase-space formalisms of quantum
mechanics, with also a more intrinsic and coordinate independent notation,
in order to discuss some of its most outstanding and up-to-date
applications: namely the theory of modular variables and the concept of weak
values developed by Yakir Aharonov and his many collaborators in the last
forty years. The theory of non-relativistic Quantum Mechanics was created,
or discovered, back in the 1920's mainly by Heisenberg and Schr\"{o}dinger,
but it is fair enough to say that a more modern and unified approach to the
subject was introduced by Dirac and Jordan with their (intrinsic)
Transformation Theory. In his famous text book on Quantum Mechanics \cite%
{Dirac}, Dirac introduced his well-known bra and ket notation and a view
that even Einstein, who was, as is well known, very critical towards the
general quantum physical world-view, considered the most elegant
presentation of the theory at that time \cite{Farmelo}. One characteristic
of this formulation is that the observables of position and momentum are
truly treated equally so that an intrinsic phase-space approach seems a
natural course to be taken. In fact, we may distinguish at least two
different quantum mechanical approaches to the structure of the quantum
phase space: The Weyl-Wigner formalism and the advent of the theory of
coherent states. 
We have used these phase-space formalisms of Quantum Mechanics in order to describe two major insights due to Aharonov
and his many collaborators over the last four decades: The concept of 
\textit{weak values} and those of \textit{modular variables}. The quantum
mechanical phase space approach is particularly well suited to decribe the
concept of weak value as complex number-valued element of reality where both
the real and imaginary parts have equal physical status in the theory. The
modular variable is a concept that is slightly more difficult to grasp
mathematically. It gives rise to dynamical non-local effects which seem to
be intrinsically different from the kinematic non-local effects of entangled
states like EPR pairs of particles.

Though there is no doubts about the fact that the structure of modular
variables is much easier to understand intuitively within the Heisenberg
picture, the problem of \textit{how} to describe the finite dimensional
Hilbert space where the modular variables actually ``live in" is a quite different and
subtle issue. It is very important to notice that the Schwinger-like
approach that we present here for the construction of the modular variable
qubits (or qudits) is different from the original path chosen by Aharonov
and collaborators. In their original approach, the whole state space is a \textit{%
direct sum} between the modular variable qudit space and the
``rest of the quantum space" while in the Schwinger-based
approach the whole state space is the \textit{tensor product} between the qudit
space and \textquotedblleft the rest of the quantum space". This difference
is \textit{not} just one of academic nature. One must remember the important
fact that the rule of how one \textit{compounds} two distinct quantum
systems is described mathematically by the tensor product between each
subsystem \textit{instead} of their direct sum. This important fact is what
is behind one of the most fundamental differences between classical and
quantum physics, the phenomenon of quantum \textit{entanglement} which is
the basis of modern quantum information theory. One may speculate that this
fact could lead to some possible experimental procedure in order to decide
this important issue. One could envisage the use of the quantum mechanical
``pseudo-degrees of freedom" (as the periodicity of a lattice of slits) as a 
\textit{new resource} for manipulation and storage of qubits with possible
new applications in the field of quantum information. We also hope that this
may help to shed some light on the physical differences between dynamic
non-locality and that of kinematic non-locality in quantum physics.

%

\begin{acknowledgements}
A. C. Lobo and C. A. Ribeiro acknowledge financial support from Conselho Nacional de Desenvolvimento
Científico e Tecnológico (CNPq), and we also gratefully acknowledge the support of Chapman University and the Institute for Quantum
Studies. This work has been supported in part by the Israel Science Foundation Grant No. 1125/10.
\end{acknowledgements}



\end{document}